\documentclass[a4paper]{article}
\usepackage[margin=1in]{geometry}  
\usepackage[english]{babel}
\usepackage[utf8x]{inputenc}
\usepackage[T1]{fontenc}
\usepackage{parskip}
\usepackage[dvipsnames]{xcolor}
\usepackage{graphicx}
\usepackage{subcaption}
\usepackage{longtable} 
\usepackage{multirow}
\usepackage{listings}
\usepackage{makecell}
\usepackage{array}
\usepackage{float}
\usepackage{dsfont}
\usepackage{rotating}
\usepackage{booktabs}
\usepackage{enumerate}
\usepackage{tikz}
\usetikzlibrary{chains}
\usetikzlibrary{arrows}
\usepackage{pgf}
\usepackage{diagbox}
\usepackage{amsmath}
\usepackage{amssymb}
\usepackage{amsthm}
\usepackage{bm}
\usepackage{mathtools}
\usepackage{mathrsfs}
\usepackage{dsfont} 
\usepackage{natbib}
\usepackage[
  colorlinks,
  citecolor=blue,
  linkcolor=red,
  anchorcolor=red,
  urlcolor=blue
]{hyperref}
\mathtoolsset{showonlyrefs}
\usepackage{authblk}
\usepackage[ruled]{algorithm2e}
\usepackage[margin=1in]{geometry}  
\usepackage[english]{babel}
\usepackage[utf8x]{inputenc}
\usepackage[T1]{fontenc}
\usepackage{parskip}
\theoremstyle{plain}

\newtheorem{definition}{Definition}
\newtheorem{theorem}{Theorem}
\newtheorem{lemma}{Lemma}

\newtheorem{remark}{Remark}

\newcommand{\RR}{\mathbb{R}}
\newcommand{\EE}{\mathbb{E}}
\newcommand{\PP}{\mathbb{P}}
\newcommand{\ind}{\mathds{1}}

\newcommand{\indep}{\rotatebox[origin=c]{90}{$\models$}}

\usepackage{xspace}

\newcommand{\given}{{\,|\,}}
\newcommand{\biggiven}{\,\big{|}\,}
\newcommand{\Biggiven}{\,\Big{|}\,}
\newcommand{\bigggiven}{\,\bigg{|}\,}

\def\@#1\@{\begin{align}#1\end{align}}
\def\$#1\${\begin{align*}#1\end{align*}}

\definecolor{myblue}{rgb}{.8, .8, 1}
\definecolor{mathblue}{rgb}{0.2472, 0.24, 0.6} 
\definecolor{mathred}{rgb}{0.6, 0.24, 0.442893}
\definecolor{mathyellow}{rgb}{0.6, 0.547014, 0.24}

\newcommand{\cA}{{\mathcal{A}}}

\newcommand{\cC}{{\mathcal{C}}}
\newcommand{\cD}{{\mathcal{D}}}
\newcommand{\cE}{{\mathcal{E}}}

\newcommand{\cI}{{\mathcal{I}}}

\newcommand{\cM}{{\mathcal{M}}}
\newcommand{\cN}{{\mathcal{N}}}

\newcommand{\cP}{{\mathcal{P}}}

\newcommand{\cT}{{\mathcal{T}}}
\newcommand{\cU}{{\mathcal{U}}}

\newcommand{\cX}{{\mathcal{X}}}
\newcommand{\cY}{{\mathcal{Y}}}

\newcommand{\bart}{{\bar{t}}}

\newcommand{\tr}{\textnormal{tr}}
\newcommand{\test}{\textnormal{test}}
\newcommand{\quant}{\text{Quantile}}
\renewcommand{\hat}{\widehat}
\newcommand{\bern}{\text{Bern}}
\newcommand{\iid}{\text{i.i.d.}}
\newcommand{\eps}{\varepsilon}
\newcommand{\tw}{\widetilde{w}}
\newcommand{\esterr}{\textnormal{EstErr}}
\newcommand{\dd}{\textnormal{d}}

\title{Not all distributional shifts are equal:\\ 
Fine-grained robust conformal inference}
\author[1]{Jiahao Ai}
\author[2]{Zhimei Ren}
\affil[1]{School of Mathematical Sciences, Peking University}
\affil[2]{Department of Statistics and Data Science, University 
of Pennsylvania}
\date{\today}

\begin{document}

\maketitle
\begin{abstract}
    We introduce a fine-grained framework for uncertainty quantification of 
    predictive models under distributional shifts.
    This framework distinguishes the shift in covariate distributions from 
    that in the conditional relationship between the outcome ($Y$) and the covariates ($X$). 
    We propose 
    to reweight the training samples to adjust for an identifiable  
    shift in covariate distribution while protecting against the  worst-case 
    conditional distribution shift bounded in an $f$-divergence ball. 
    Based on ideas from conformal 
    inference and distributionally robust learning, 
    we present an algorithm that outputs (approximately) valid and efficient 
    prediction intervals in the presence of distributional shifts. 
    As a use case, we apply the framework to 
    sensitivity analysis of individual treatment effects with hidden confounding.
    The proposed methods are evaluated in simulations and 
    four real data applications, demonstrating superior robustness and 
    efficiency compared with existing benchmarks.
\end{abstract}

\section{Introduction}
It has been widely observed that the performance of 
predictive models falls short of expectation
when generalized to a population whose distribution 
differs from that of the training data 
(see e.g.,~\citet{recht2019imagenet,miller2020effect,
wong2021external,namkoong2023diagnosing,liu2023need} 
and the references therein). As predictive models 
are increasingly employed in high-stakes settings,
it is imperative to accompany the  predicted outcomes  
with calibrated uncertainty quantification
when deploying a model to new environments.
A widely adopted approach to uncertainty quantification is to 
provide a prediction set that contains the true outcome 
with high probability.
The prediction set informs the confidence 
we have in the predicted outcome.

Among the tools for constructing prediction sets,
conformal prediction (CP)~\citep{vovk2005algorithmic} is 
an attractive framework that generates valid prediction sets 
that are guaranteed 
to include the true outcome 
with pre-specified probability.
The validity of CP holds for {\em any} predictive model, as long as 
the training and test data are exchangeable, e.g., when they are 
identically independently distributed (i.i.d.). 
In the presence of distributional shifts, however, the exchangeability/i.i.d.~assumption
breaks, and CP no longer delivers valid prediction sets. 
To address this challenge, prior work~\citep{cauchois2023robust} proposes a robust CP 
method that outputs prediction sets that are valid when the target  
distribution ranges within a neighborhood of the 
training distribution. To be more specific,
let $X \in \cX$ denote the covariates and $Y \in \cY$ the outcome/response.
Consider a training set of $n$ samples $(X_i,Y_i)\stackrel{\iid}{\sim}P_{X,Y}$ 
and an independent test unit $(X_{n+1},Y_{n+1}) \sim Q_{X,Y}$, where we 
only get to observe $X_{n+1}$ and wish to predict $Y_{n+1}$.~\citet{cauchois2023robust} 
assumes that the $f$-divergence between $Q_{X,Y}$ and 
$P_{X,Y}$ is bounded by a parameter $\rho$, and provides
prediction sets that ensure guarantees even for the worst-case $Q_{X,Y}$.

The method of~\citet{cauchois2023robust} provides robustness 
against the worst-case {\em joint} distributional shift of $(X,Y)$, 
but no distinction is made between the covariate shift and the
conditional distributional shift. As pointed out by a recent line of research,  
different types of distributional shifts appear in different
tasks and result in different consequences~\citep{mu2022factored, 
namkoong2023diagnosing,jin2023diagnosing,liu2023need}.
Without separating the sources of distributional shifts and 
taking specialized treatment, we will show that the joint modeling approach
of~\citet{cauchois2023robust} can be overly conservative in practice.
In this work, we take a closer look at distributional shift, 
and provide a fine-grained robust predictive inference approach with 
improved efficiency.
\subsection{Decomposing the distributional shifts}
We decompose the distributional shifts into two types:
\begin{enumerate}
\item [(1)] {\em The covariate shift:} the marginal distribution of $X$
is different in the training and target environment. For example, 
the age/gender structure in the new environment differs from that in 
the training environment. 

\item [(2)] {\em The $Y\given X$ shift:}
the conditional relationship between the outcome and the 
the covariate is different in the training and target 
environment. This could happen when there are unobserved 
confounders, or when the training and target data are 
collected from different periods and the conditional relationship
varies over time.
\end{enumerate}
The above two types of distributional shifts are 
different in nature --- for one thing, the former type of 
distributional shift is {\em identifiable} but the latter is not. 
In most cases, the distributional shift is a mixture of the two. 
Instead of guarding against the worst-case joint distributional shift, 
we propose to tease apart the two types of shifts,
reweighting the training samples according to the estimated covariate shift
and adjusting the confidence level to account for the worst-case 
$Y\given X$ shift.

Specifically, we assume the $Y\given X$ shift to be 
bounded in the $f$-divergence, i.e., $D_f(Q_{Y\given X} \,\|\, 
P_{Y\given X}) \le \rho$, but posit no constraints 
on the covariate shift. For such distributional shifts, 
our proposed method aims to construct
a prediction interval $\hat{C}_{f,\rho}(X_{n+1})$ with the 
training data, such that it covers the true outcome with 
high probability under the target distribution.


\subsection{Our contributions}
This work introduces a new framework for 
calibrated uncertainty quantification in the 
presence of distributional shifts.
Toward this end, we make the following contributions:
\begin{enumerate}
\item [(1)] We present {\em Weighted Robust Conformal 
Prediction (WRCP)}, which treats the covariate shift 
and the $Y\given X$ shift differently. It (approximately) 
achieves the desired coverage under the proposed framework, 
with the miscoverage rate determined by the estimation 
error of the covariate likelihood ratio $\dd Q_X/\dd P_X$.
\item [(2)] In the case when estimating 
the covariate shift is challenging (e.g., 
when $X$ is high-dimensional), we propose 
a debiased variant of WRCP, namely D-WRCP, which enjoys the 
double-robustness property --- its
miscoverage rate depends on the product of the estimation error of 
$\dd Q_X /\dd P_X$ and that of conditional quantiles  
of the residuals from predicting the outcomes.
\item [(3)] As a special example, we show that our proposed methods
can be adapted to conducting sensitivity analysis for individual 
treatment effects (ITEs) under  the 
$f$-sensitivity model~\citep{jin2022sensitivity}.
\item [(4)] We empirically evaluate the proposed methods 
in simulations and four real data applications, demonstrating 
their validity and improved efficiency.
\end{enumerate}

\subsection{Related literature}
\paragraph{Conformal prediction beyond exchangeability.}
With exchangeable/i.i.d.~data, there is a long list of works on the theoretical 
property, efficient implementation and application of conformal 
prediction (see e.g.,~\citet{vovk2005algorithmic,papadopoulos2002inductive,
lei2018distribution,romano2019conformalized,foygel2021limits,angelopoulos2023conformal}). 

Beyond exchangeability,~\citet{tibshirani2019conformal,park2022pac} 
consider the pure covariate shift setting, with the former focusing
on the marginal coverage guarantee and the latter the training-conditional guarantee;
also under the pure covariate shift setting,~\citet{qiu2022distribution,yang2022doubly}
builds upon semi-parametric theory to develop more efficient CP methods 
with asymptotic coverage guarantees.
The de-biased version of our proposal draws inspiration
from these two works, and we generalize them to the specific 
distributional shift model under consideration.~\citet{podkopaev2021distribution,si2023distributionally} 
tackles the label shift setting, where the marginal distribution of 
$Y$ is subject to changes but $X \given Y$ remains invariant in the training and target distribution.  
The work of~\citet{barber2023conformal} addresses a general form of distribution shift by up-weighting 
training points whose distribution is closer to that of the target distribution (the weights need to be independent 
of the data).

As mentioned earlier,~\cite{cauchois2023robust} is concerned with robust CP 
against the worst-case joint shift in $(X,Y)$.~\citet{gendler2021adversarially,ghosh2023probabilistically}
investigate the robustness of CP under adversarial attacks.
Another two closely related works 
on robust CP are~\citet{jin2023sensitivity,yin2022conformal}, which study sensitivity analysis 
of ITEs under the marginal $\Gamma$-selection model~\citep{tan2006distributional};
the type of distributional shift (caused by hidden confounding) puts no requirements 
on the covariate shift and assumes that the shift in $Y\given X$ is  uniformly bounded by constants, i.e.,
$1/\Gamma \le \frac{\dd Q_{Y \given X}}{\dd P_{Y \given X}} \le \Gamma$.
Compared with our model that assumes the $Y\given X$ shift to be bounded {\em on average} 
(the $f$-divergence takes the expectation over $Y$), the point-wise bound requires the 
{\em maximum} shift to be bounded, which can sometimes be conservative in practice (see more 
discussion and examples in~\citet{jin2022sensitivity}).

\paragraph{Distributionally robust learning.}
Distributionally robust learning studies the broad topic of learning from data 
with guarantees under the worst-case distributional shift within a specified set of distributions.
Typical tasks in this field includes 
parameter estimation~\citep{shafieezadeh2015distributionally,blanchet2019quantifying,duchi2021learning,duchi2023distributionally},
policy learning~\citep{si2023distributionally,mu2022factored,zhang2023optimal}, among others. 
In particular,~\citet{mu2022factored} proposes learning a robust policy by 
separately considering covariate shifts and $Y \given X$ shifts, echoing the proposal 
in this paper. Compared with the existing literature, 
our work takes a different angle by studying the
uncertainty quantification problem under distributional shifts.

\paragraph{Sensitivity analysis.}
In causal inference, distributional shifts can arise due to unobserved confounders, and 
sensitivity analysis is a standard tool for assessing the robustness of causal effect 
estimates under such shifts. 
Under the aforementioned (marginal)
$\Gamma$-selection model~\citep{rosenbaum1987sensitivity,tan2006distributional}, 
a line of papers~\citep{zhao2019sensitivity, yadlowsky2018bounds,kallus2020confounding,sahoo2022learning}
study the estimation of the average treatment effect (ATE) or the policy values, 
and the work of~\citet{kallus2021minimax,lei2023policy} 
consider learning the optimal policy.
Recently,~\cite{jin2022sensitivity} proposes the $f$-sensitivity model,
and discusses how to estimate the ATE under the model. We shall show later in this paper 
that the distributional shift under the $f$-sensitivity model fits exactly in our framework, and hence 
our proposed method can be adopted there for the uncertainty quantification of ITEs.

\section{Problem setup}
Consider a training data set $\cD_\tr = \{(X_i,Y_i)\}_{i=1}^n$, 
where $(X_i,Y_i) \stackrel{\text{i.i.d.}}{\sim} P_{X,Y}$. For a test 
unit $(X_{n+1},Y_{n+1}) \sim Q_{X,Y}$, for which only the covariate is observed, 
we aim at using $\cD_\tr$ to construct an interval $\hat{C}(X_{n+1})$ such that  
\@\label{eq:marg_pi}
\PP_{(X_{n+1},Y_{n+1}) 
\sim Q_{X,Y}}\big(Y_{n+1} \in \hat{C}(X_{n+1})\big) \ge 1-\alpha,
\@
where the probability is taken over the randomness of $(X_i,Y_i)\stackrel{\iid}{\sim}P_{X,Y}$
and $(X_{n+1},Y_{n+1})\sim Q_{X,Y}$, and $\alpha \in (0,1)$ 
is the  pre-specified mis-coverage level.

Let $s: \cX \times \cY \mapsto \RR$ denote a score function, and 
we define for each $i \in [n] = \{1,2,\ldots,n\}$ 
the nonconformity score $S_i = s(X_i,Y_i)$. 
For example, when $\hat{\mu}(x)$ is a fitted function of the conditional mean of $Y \given X$, 
one can take  $s(x,y) = |y - \hat{\mu}(x)|$.\footnote{Strictly, we should write the 
score function as $s(x,y;\hat{\mu})$ as it also depends on $\hat{\mu}$. For notational 
simplicity, we suppress the dependence on $\hat{\mu}$ (or other predictive functions) 
in the score function when the 
context is clear. } For other types of nonconformity scores, 
see also~\citet{romano2019conformalized,chernozhukov2021distributional,guan2023localized,gupta2022nested}.
In order to achieve~\eqref{eq:marg_pi}, it suffices to find 
(an upper bound of) the $(1-\alpha)$-th quantile of $s(X_{n+1},Y_{n+1})$ under $Q_{X,Y}$. 



\subsection{Characterizing the distributional shifts}
As introduced earlier, the distributional shift between $P_{X,Y}$ 
and $Q_{X,Y}$ can be originating from two sources:
(1) the difference 
between $P_X$ and $Q_X$ and
(2) 
the difference between $P_{Y \given X}$ and $Q_{Y \given X}$.
The two types of distribution shifts are 
different in nature: often the 
covariate shift is observable and estimable --- since we have access 
to the covariates in the test set --- while the $Y\given X$
shift is not identifiable. Based on this observation,  
we propose distinct treatments to these two types of 
distributional shift.

The covariate shift is represented by the 
likelihood ratio $w(x) = \frac{dQ_X}{dP_X}(x)$. 
We do not posit any assumption on $w(x)$ (except that 
$Q_X$ is absolutely continuous with respect to $P_X$), and shall use 
data to estimate this quantity. 
For the conditional distributional shift, 
we assume that the target distribution $Q_{Y\given X}$ falls 
within a ``neighborhood ball'' of $P_{Y\given X}$, whose radius
is controlled by a parameter $\rho$. The neighborhood ball
is formalized by the $f$-divergence. 

\begin{definition}[$f$-divergence]
Let $P$ and $Q$ be two probability distributions over a space $\Omega$ 
such that $P$ is absolutely continuous with respect to $Q$. 
For a convex function $f$ such that $f(1) = 0$, 
the $f$-divergence of $P$ from $Q$ is defined as 
$D_f(P \,\|\, Q) = \EE_Q [f(dP/dQ)]$, where $dP/dQ$
is the Radon-Nikodym derivative.
\end{definition}
Throughout, we assume $f$ to be 
closed and convex, with $f(1) = 0$ and $f(x)<+\infty$ for 
$x>0$. Common choices of $f$ include $f(x) = x\log x$,
which yields the Kullback–Leibler (KL) divergence, 
$f(x) = \frac{1}{2}|x-1|$ that yields the total variation (TV) distance, 
and  $f(x) = (x-1)^2$ that yields the Pearson $\chi^2$-divergence.

With the target conditional distribution of 
$Y \given X$ satisfying $D_f(Q_{Y\given X = x} \,\|\, 
P_{Y \given X=x}) \le \rho$, for $P_X$-almost all $x$,
we can define the set of possible 
$Q_{X,Y}$ as
\@\label{eq:iden_set}
\cP(\rho; P) \,:=\, 
\big\{Q \text{ s.t. }(X_{n+1},Y_{n+1})\sim Q: D_f\big(Q_{Y\given X = x} \,\|\,P_{Y\given X=x}\big)\leq \rho, 
\text{ for }P_X\text{-almost all }x\big\}.
\@
In what follows, we shall refer to $\cP(\rho;P)$ as 
the identification set. When $Q \in \cP(\rho; P)$, the
task in~\eqref{eq:marg_pi} can be equivalently written as 
\$ 
\inf_{Q \in \cP(\rho;P)} \PP_{(X_{n+1},Y_{n+1})\sim Q}\big(Y_{n+1}
\in \hat{C}(X_{n+1})\big) \ge 1-\alpha.
\$

\subsection{Split conformal prediction}
When $P_{X,Y} = Q_{X,Y}$, the method of conformal inference 
offers an elegant solution for finding the quantile of $s(X_{n+1},Y_{n+1})$
by leveraging the exchangeability 
among $\{(X_i,Y_i)\}_{i=1}^{n+1}$. In particular, the split 
conformal inference~\citep{vovk2005algorithmic,papadopoulos2002inductive} 
is a computationally efficient variant of conformal inference that
begins by randomly splitting the training data into two 
folds, $\cD_\tr^{(0)}$ and $\cD_\tr^{(1)}$, where $n_0 = |\cD_\tr^{(0)}|$
and $n_1 = |\cD_\tr^{(1)}|$.
It then uses $\cD_\tr^{(0)}$
for fitting the prediction function 
$\hat{\mu}: \cX \mapsto \RR$ and $\cD_\tr^{(1)}$ for obtaining the 
estimated quantile. The prediction interval takes the form 
\@ \label{eq:conf_pi}
\hat{C}(X_{n+1}) = \Big\{y \in \RR: s(X_{n+1},y)\le  
\quant\Big(1-\alpha,
\{S_i\}_{i \in \cD_\tr^{(1)}} \cup \{\infty\}\Big)\Big\},
\@
where $\quant(\beta,\{Z_i\}_{i=1}^n)$ denotes the $\lceil n\beta \rceil$-th smallest 
element among $Z_1,Z_2,\ldots,Z_n$. The prediction interval~\eqref{eq:conf_pi}
guarantees that $\PP(Y_{n+1} \in \hat{C}(X_{n+1})) \ge 1-\alpha$ 
when $P_{X,Y} = Q_{X,Y}$ without any additional assumptions~\citep{vovk2005algorithmic}; 
if the ties among the nonconformity scores happen with probability zero, 
then the coverage is also tight~\citep{lei2018distribution}, i.e., 
$\PP(Y_{n+1} \in \hat{C}(X_{n+1})) \le 1-\alpha +
1/n_1$.



\section{Methodology}
In this section, we describe how to 
generalize (split) conformal inference 
to efficiently handle distributional shift.
To start, we fix the radius of the identification set $\rho >0$. 
As in the standard split conformal prediction, we 
start by splitting the training set into two folds, $\cD_{\tr}^{(0)}$
and $\cD_{\tr}^{(1)}$. The fitting fold $\cD_\tr^{(0)}$ is used for 
fitting the prediction function $\hat{\mu}$ (or other functions depending on the 
type of nonconformity score). The calibration fold $\cD_\tr^{(1)}$
is devoted to finding the largest quantile of $S_{n+1}$ for $Q \in \cP(\rho;P)$.
To this end, we follow~\citet{cauchois2023robust} and define
\begin{align}
g_{f,\rho}(\beta) \,:=\, \inf\bigg\{z\in [0,1]:\beta f\Big(\frac{z}{\beta}\Big)
+(1-\beta)f\Big(\frac{1-z}{1-\beta}\Big)\leq \rho\bigg\},
\end{align}
and its inverse
\begin{align}
g_{f,\rho}^{-1}(\tau) \,:=\, \sup\big\{\beta \in [0,1]:g_{f,\rho}(\beta)\leq \tau\big\}.
\end{align}
Recall that $w(x) = \frac{dQ_X}{dP_X}(x)$.
We construct our prediction set as
\begin{align}
\label{eq:pi}
& \hat{C}_{f,\rho}(x)=
\bigg\{y\in \RR: s(x,y)\leq \text{Quantile} 
\Big(g_{f,\rho}^{-1}(1-\alpha),
\sum_{i \in \cD_\tr^{(1)}} 
p_i(x)\delta_{S_i} + p_{n+1}(x)\delta_\infty \Big) \bigg\},\\
& \text{where }p_i(x) = \frac{w(X_i)}{\sum_{j \in \cD_\tr^{(1)}} w(X_j) 
+ w(x)}, 
\text{ and }
p_{n+1}(x) = \frac{w(x)}{\sum_{j \in \cD_\tr^{(1)}} w(X_j) 
+ w(x)}.
\end{align}
In words, we upper bound the $(1-\alpha)$-th quantile 
of $S_{n+1}$ under $Q$ by a weighted quantile under $P$ at 
a slightly inflated level. The validity of $\hat{C}_{f,\rho}(X_{n+1})$ 
is formalized by Theorem~\ref{thm:cov_known_shift}, 
whose proof is deferred to Appendix~\ref{appx:proof_cov_known_shift}.

\begin{theorem}[Prediction interval with known covariate shift]
\label{thm:cov_known_shift}
Assume the training data $\{(X_i,Y_i)\}_{i=1}^n 
\stackrel{\text{i.i.d.}}{\sim} P_{X,Y}$ 
and  $(X_{n+1}, Y_{n+1})\sim Q_{X, Y}$ 
is independent of $\{(X_i,Y_i)\}_{i=1}^n$.
Assume that $Q$ is absolutely continuously continuous 
with respect to $P$, and denote $w(x) = \frac{dQ_X}{dP_X}(x)$.
For $\alpha \in (0,1)$, the prediction set $\hat{C}_{f,\rho}(X_{n+1})$
defined in~\eqref{eq:pi}
satisfies that
\begin{align}
    \mathbb{P}\big(Y_{n+1}\in \hat{C}_{f,\rho}(X_{n+1})\big)\geq 
    g_{f,\rho}\big(g^{-1}_{f,\rho}(1-\alpha)\big).
\end{align}
Furthermore, if $g_{f,\rho}(1) \ge 1-\alpha$, then 
\$ 
\mathbb{P}\big(Y_{n+1}\in \hat{C}_{f,\rho}(X_{n+1})\big)\geq 
1-\alpha.
\$
\end{theorem}
The condition that $g_{f,\rho}(1) \ge 1-\alpha$ holds for 
the KL divergence and the $\chi^2$ distance for any choice of $\rho,\alpha>0$; 
it holds for the TV distance for  $\alpha \ge \rho/2$. 
Two remarks are in order.
\begin{remark}
As shown in~\citet[Lemma A.1]{cauchois2023robust}, 
$g_{f,\rho}(\beta)$ is non-decreasing in $\beta$, 
which allows for efficient computation of $g_{f,\rho}^{-1}(\tau)$.
For example, by binary search, we can get an estimate of
${g}_{f,\rho}^{-1}(\tau)$ with error $\varepsilon$ within 
$O(\log({(1-\tau)}/{\epsilon}))$ runs.
\end{remark}

\begin{remark}
When there is no distributional shift in $Y\given X$, i.e., $\rho = 0$, 
our method recovers split weighted conformal 
prediction~\citep{tibshirani2019conformal}; when there is no covariate shift, 
i.e., $w(x) \equiv 1$, it recovers the method of~\citet{cauchois2023robust}.
Our procedure is therefore a generalization of both methods.
\end{remark}

We now have a general recipe for handling distributional shifts
in $\cP(\rho;P)$. So far the recipe requires that the covariate shift
$w(x)$ to be specified a priori --- this 
may be the case where the covariate shift is induced by 
a covariate-based selection rule that is known to the experimenter ---
but more often, we do not know the exact form of $w(x)$.
The following section discusses how to estimate $w(x)$ with data 
and how the coverage depends on the estimation quality.

\subsection{Estimating the covariate shift}
Consider a common scenario in prediction tasks: 
there are multiple test units 
denoted by $ \cD_{\test} = \{(X_{n+j},Y_{n+j})\}_{j=1}^m$,
where $(X_{n+j},Y_{n+j}) \stackrel{\text{i.i.d.}}{\sim} Q_{X,Y}$.
For each $j \in [m]$, we aim to construct a prediction interval
$\hat{C}_{f,\rho,n+j}(X_{n+j})$ satisfying~\eqref{eq:marg_pi}.

The multiple test units allow us to estimate $w(x)$. 
In particular, we adopt the estimation approach introduced in~\citet{tibshirani2019conformal},
where we first randomly split $\cD_{\test}$
into two folds: $\cD_{\test}^{(0)}$ and $\cD_{\test}^{(1)}$, 
indexed by $\cI_{\test}^{(0)}$ and $\cI_{\test}^{(1)}$, respectively.
Without loss of generality, assume that $n+j \in \cI_\test^{(1)}$.
Recall that the training set $\cD_\tr$ is also divided into 
$\cD_\tr^{(0)}$ and $\cD_\tr^{(1)}$. 
We set aside $\cD_\tr^{(0)}\cup \cD_\test^{(0)}$
for estimating $w(\cdot)$. Let $A$ be a binary variable
indicating whether the sample is from the training 
set or the test set, i.e., $A_i = 0$ for $i \in \cI_\tr^{(0)}$
and $A_i = 1$ for $i \in \cI_\test^{(0)}$.
For $i \in \cI_\tr^{(0)} \cup \cI_\test^{(0)}$, by Bayes' rule, 
\$ 
\frac{\PP(A_i = 1 \given X_i = x)}{\PP(A_i = 0 \given X_i = x)}
= \frac{dQ_X}{dP_X}(x) \cdot \frac{\PP(A_i = 1)}{\PP(A_i = 0)}
\propto w(x).
\$
The above tells us that the likelihood ratio $w(x)$ can be estimated
by training a classifier on $\cD_\tr^{(0)}\cup \cD_\test^{(0)}$:
once we  obtain $\hat{\PP}(A = 1\given X = x)$, 
we can let $\hat{w}(x) = \frac{\hat{\PP}(A=1\given X=x)}
{1 - \hat{\PP}(A = 1\given X = x)}$  --- this is an estimator 
for $w(x)$ (up to constants).
We then construct the prediction interval by replacing $w(x)$ with
$\hat{w}(x)$ in~\eqref{eq:pi}.
The complete procedure is summarized in Algorithm~\ref{alg:wrcp}, 
and the following theorem provides the coverage guarantee 
when the estimated $w(x)$ is used.

\begin{theorem}\label{thm:cov_est}
Under the same assumptions of Theorem~\ref{thm:cov_known_shift}, 
suppose that $\mathbb{E}_{X\sim P_X}\big[\hat{w}^{(k)}(X)]<\infty$, 
for $k \in \{0,1\}$.
Then for any $k \in \{0,1\}$ and any $n+j \in \cI_\test^{(k)}$, 
the prediction set of Algorithm~\ref{alg:wrcp} satisfies
\begin{align}
\mathbb{P}\big(Y_{n+j} \in \hat{C}_{f,\rho,n+j}(X_{n+j})\big)
& \ge g_{f,\rho}\big(g^{-1}_{f,\rho}(1-\alpha)\big) 
- \frac{1}{2}g_{f,\rho}'\big(g^{-1}_{f,\rho}(1-\alpha)\big)
\cdot \mathbb{E}_{X\sim P_X}\bigg[\Big|\frac{\hat{w}^{(k)}(X)}
{\EE[\hat{w}^{(k)}(X)]}-w(X)\Big|\bigg],
\end{align}
where $g'_{f,\rho}$ is the left derivative of $g_{f,\rho}$.
Furthermore, if $g_{f,\rho}(1) \ge 1-\alpha$, then 
\$ 
\mathbb{P}\big(Y_{n+j}\in \hat{C}_{f,\rho,n+j}(X_{n+j})\big)\geq 
1-\alpha
- \frac{1}{2}g_{f,\rho}'\big(g^{-1}_{f,\rho}(1-\alpha)\big)
\cdot \mathbb{E}_{X\sim P_X}\bigg[\Big|\frac{\hat{w}^{(k)}(X)}
{\EE[\hat{w}^{(k)}(X)]}-w(X)\Big|\bigg].
\$
\end{theorem}
The proof of Theorem~\ref{thm:cov_est} is based on the coupling 
technique used in~\citet{lei2021conformal}, and can be found in 
Appendix~\ref{appx:proof_cov_est}.

\begin{remark}
If the number of test units $m$ is small, one can replace 
$\cD_\test^{(1-k)}$ with  $\cD_\test\backslash \{X_{n+j}\}$ when estimating $w(x)$, 
i.e., train the classifier on $\cD_\tr^{(0)} \cup 
\cD_\test \backslash \{X_{n+j}\}$. This approach can
improve the accuracy of the classifier but may be
computationally intensive when $m$ is large, so we present 
the sample-splitting version for simplicity.
\end{remark}

With estimated $w(x)$, Theorem~\ref{thm:cov_est} 
suggests that the miscoverage rate inflation 
depends on the estimation error of $\hat{w}(x)$.
In general, when $x$ is low-dimensional, we can obtain
a relatively accurate estimator of $w(x)$, and the resulting 
prediction interval is approximately valid. In other situations where 
high-dimensional covariates are present, estimating 
$w(x)$ can be challenging. To handle this issue, we propose 
an alternative method that leverages the debiasing technique 
to construct efficient prediction intervals. We present it in detail 
in the following section.

\begin{algorithm}[htbp]
\caption{Weighted robust conformal prediction (WRCP)}
\label{alg:wrcp}
\KwIn{Training set $\cD_\tr = \{(X_i,Y_i)\}_{i=1}^n$; 
test data $\cD_\test = \{X_{n+j}\}_{j=1}^m$;  
regression algorithm $\mathcal{A}$; 
classification algorithm $\cC$;
target miscoverage level $\alpha \in (0,1)$; 
score function $s(x,y;\mu)$;
robust parameter $\rho$.\\
\textbf{Optional input:} likelihood ratio function $w(x)$. }
\vskip 1em
Randomly split $\cD_\tr$ into 
two disjoint subsets of equal sizes, $\cD^{(0)}_\tr$ 
and  $\cD^{(1)}_\tr$, indexed by $\cI_\tr^{(0)}$ and 
$\cI_\tr^{(1)}$, respectively\;
\vskip 0.2em 
Apply $\mathcal{A}$ to $\cD^{(0)}_\tr$ and 
obtain the prediction function: 
$\hat \mu \leftarrow \mathcal{A}(\cD^{(0)}_\tr)$\;
\vskip 0.2em
Compute the nonconformity score 
$S_i=s(X_i, Y_i)$ for $ i \in \cI^{(1)}_\tr$\;
\vskip 0.2em
\uIf{$w(x)$ exists}{
\For{$j = 1,\ldots,m$}{
Construct 
$\hat{C}_{f,\rho,n+j}(X_{n+j})$ according to~\eqref{eq:pi}\;
}
}\Else{
Split $\cD_\test$ into two disjoint subsets of 
equal sizes, $\cD_\test^{(0)}$ and $\cD_\test^{(1)}$, 
indexed by $\cI_\test^{(0)}$ and $\cI_\test^{(1)}$, respectively\;
\vskip 0.2em
\For{$k=0,1$}{
Train a classifier:
$\hat{\PP}^{(k)}(A = 1 \given X=x) \leftarrow 
\cC(\cD_\tr^{(0)},\cD_\test^{(1-k)})$\;
Construct the estimator  
$\hat{w}^{(k)}(x) \leftarrow \frac{\hat{\PP}^{(k)}(A = 1 \given X = x)}
{1-\hat{\PP}^{(k)}(A=1\given X=x)}$\;
\For{$\ell \in \cI_{\textnormal{test}}^{(k)}$}{
Construct 
$\hat{C}_{f,\rho,\ell}(X_{\ell})$ according to~\eqref{eq:pi}
with $w(x)$ replaced by $\hat{w}^{(k)}(x)$\;
}
}
}
\vskip 1em

\KwOut{Prediction sets $\{\hat{C}_{f,\rho,n+j}(X_{n+j})\}_{j\in[m]}$.}
\end{algorithm}

\subsection{Doubly robust prediction sets}
Continue focusing on the test unit $n+j \in \cI_\test^{(1)}$.
Recall that we fit $\hat{w}^{(1)}(x)$ on $\cD_\tr^{(0)}\cup \cD_\test^{(0)}$; 
we now reuse $\cD_\tr^{(0)}$ to fit the function 
$x \mapsto \EE_{P}\big[\ind\{s(X,Y) \le t \} \biggiven X  = x\big]$, 
denoting the estimator by $\hat{m}^{(1)}(x;t)$.
Since our estimand is the conditional cumulative distribution function 
(CDF), we assume the estimator $\hat{m}(x;t)$ to be bounded in $[0,1]$, 
non-decreasing in $t$, and right-continuous without loss of generality.

To motivate the doubly robust prediction set, let us take another look at  
the coverage probability under a pure covariate shift 
at a fixed threshold $t$, which can be written as
\$
& \PP_{(X,Y)\sim Q_X\times P_{Y\given X}}\big(s(X,Y) \le t \big) \\
=~&  \EE_{(X,Y) \sim Q_X \times P_{Y\given X}}\big[\ind\{s(X,Y) \le t\} \big]\\
=~& \EE_{(X,Y) \sim Q_X \times P_{Y\given X}}
\big[\big(\ind\{s(X,Y) \le t\} - \hat{m}^{(1)}(X;t)\big)\big]
+ \EE_{X\sim Q_X}\big[\hat{m}^{(1)}(X;t)\big]\\
= &\frac{\EE_{(X,Y) \sim P_{X,Y}}\big[w(X) \cdot \big( 
\ind\{S(X,Y) \le t\} - \hat{m}^{(1)}(X;t)\big)\big]}
{\EE_{X\sim P_X}[w(X)]}
+ \EE_{X\sim Q_X}\big[\hat{m}^{(1)}(X;t)\big].
\$
In the above decomposition, 
the first term can be estimated with the training data, and the second 
term with the 
test data. We therefore modify the coverage probability estimator at 
threshold $t$ to be 
\$
\hat{p}^{(1)}(t) = \frac{\sum_{i \in \cI_\tr^{(1)}}\hat{w}^{(1)}(X_i)\cdot 
\big(\ind\{S_i \le t\} - \hat{m}^{(1)}(X_i;t)\big)} 
{\sum_{i \in \cI_\tr^{(1)}} \hat{w}^{(1)}(X_i)} + 
\frac{1}{|\cI_{\test,j}^{(1)}|}\sum_{i \in \cI_{\test,j}^{(1)}} 
\hat{m}^{(1)}(X_j;t),
\$
where $\cI^{(1)}_{\test,j} = \cI_{\test}^{(1)}\backslash \{j\}$.
Note that $\hat{p}^{(1)}(t)$ is no longer monotone in $t$; to obtain the 
quantile, we consider a ``monotonized'' version of $\hat{p}^{(1)}(t)$. The specific prediction interval is then constructed as
\@ 
\label{eq:dr_pi}
\hat{C}^{\text{DR}}_{f,\rho, n+j}(X_{n+j}) = \{y: s(X_{n+j},y) \le \hat{q}\}, \text{ where }
\hat{q} = \inf\Big\{t \in \RR: \inf_{t'\ge t} 
\hat{p}^{(1)}(t') \ge g^{-1}_{f,\rho}(1-\alpha)\Big\}. 
\@

\begin{algorithm}[htbp]
\caption{Debiased weighted robust conformal prediction (D-WRCP)}
\label{alg:dr_wrcp}
\KwIn{Training set $\cD_\tr = \{(X_i,Y_i)\}_{i=1}^n$; 
test data $\cD_\test = \{X_{n+j}\}_{j=1}^m$\; 
regression algorithm $\mathcal{A}$; 
classification algorithm $\cC$;
conditional CDF fitting algorithm $\cM$\;
target miscoverage level $\alpha \in (0,1)$; 
score function $s(x,y;\mu)$;
robust parameter $\rho$.}
\vskip 1em
Randomly split $\cD_\tr$ into 
two disjoint subsets of equal sizes, $\cD^{(0)}_\tr$ 
and  $\cD^{(1)}_\tr$, indexed by $\cI_\tr^{(0)}$ and 
$\cI_\tr^{(1)}$, respectively\;
\vskip 0.2em 
Randomly split $\cD_\test$ into two disjoint subsets of 
equal sizes, $\cD_\test^{(0)}$ and $\cD_\test^{(1)}$, 
indexed by $\cI_\test^{(0)}$ and $\cI_\test^{(1)}$, respectively\;
\vskip 0.4em 
\For{$k=0,1$}{
Obtain the prediction function: 
$\hat{\mu}^{(k)} \leftarrow \mathcal{A}(\cD^{(1-k)}_\tr)$\;
\vskip 0.2em
Compute the nonconformity score 
$S_i=s(X_i, Y_i; \hat{\mu}^{(k)})$ for $ i \in \cI^{(k)}_\tr \cup \cI^{(k)}_\test$\;
\vskip 0.2em
Train a classifier
$\hat{\PP}^{(k)}(A = 1 \given X=x) \leftarrow 
\cC(\cD_\tr^{(1-k)},\cD_\test^{(1-k)})$\;
\vskip 0.2em
Construct the estimator for covariate shift 
$\hat{w}^{(k)}(x) \leftarrow \frac{\hat{\PP}^{(k)}(A = 1 \given X = x)}
{1-\hat{\PP}^{(k)}(A=1\given X=x)}$\;
\vskip 0.2em
Obtain the estimated 
conditional CDF of $S$: 
$\hat{m}^{(k)} \leftarrow \cM(\cD_\tr^{(1-k)})$\;

\vskip 0.2em
\For{$\ell \in \cI_{\textnormal{test}}^{(k)}$}{
Construct 
$\hat{C}^{\textnormal{DR}}_{f,\rho,\ell}(X_{\ell})$ 
according to~\eqref{eq:dr_pi};
}
}
\vskip 1em

\KwOut{Prediction sets $\big\{\hat{C}^{\textnormal{DR}}_{f,\rho,n+j}(X_{n+j})\big\}_{j\in[m]}$.}
\end{algorithm}

The complete procedure for constructing the doubly robust prediction sets
is described in Algorithm~\ref{alg:dr_wrcp}.
Intuitively, when $\hat{p}^{(1)}(t)$ is sufficiently close to 
$\PP_{(X_{n+j},Y_{n+j}) \sim Q_X\times P_{Y\given X}}(s(X_{n+j},Y_{n+j}) \le t)$, 
$\hat{q}$ is close to the $g^{-1}_{f,\rho}(1-\alpha)$-th quantile under 
$Q_X\times P_{Y \given X}$, thereby upper bounding the $(1-\alpha)$-th quantile 
of $S_{n+j}$ under $Q_{X,Y}$.
In the following, we let $q^*(\xi)$ be the 
$(g^{-1}_{f,\rho}(1-\alpha) - \xi)$-th quantile of 
$s(X,Y)$ under $Q_X\times P_{Y\given X}$.
The validity of $\hat{C}^{\text{DR}}_{f,\rho, n+j}(X_{n+j})$ is established  
in the following theorem.
\begin{theorem}
\label{thm:dr_pi}
For any $k \in \{0,1\}$, assume that 
\begin{enumerate}
\item [(1)]  $\hat{w}^{(k)}(x) \le w_{\max} \cdot \EE_{P_X}[\hat{w}^{(k)}(X)]$;
\item [(2)] $\hat{m}^{(k)}(x;t) \in [0,1]$ is non-decreasing and right-continuous in $t$.
\end{enumerate}
Denote the product estimation error by 
\$ 
\textnormal{EstErr}^{(k)}(t) =
\big\|m^{(k)}(X;t) - \hat{m}^{(k)}(X;t)\big\|_{L_2(P)}
\cdot \bigg\|\frac{\hat{w}^{(k)}(X)}{\EE[\hat{w}^{(k)}(X)]}
- w(X)\bigg\|_{L_2(P)}, 
\$
where $m^{(k)} := \PP_{Y \given X\sim P_{Y\given X}}(s(X,Y;\hat \mu^{(k)}) \le t \given X=x)$, 
$\|\cdot \|_{L_2(P)}$ denotes the $L_2$-norm under 
$P$, and the expectation is taken conditional on $\cD_\tr^{(1-k)}$
and $\cD_\test^{(1-k)}$.
For a unit $n+j \in \cI_\test^{(k)}$, there is 
\$  
& \PP_{(X_{n+j},Y_{n+j})\sim Q_{X,Y}}\big(Y_{n+j} \in \hat{C}_{f,\rho,n+j}^{\textnormal{DR}}(X_{n+j}) \biggiven 
\cD_\tr^{(1-k)}, \cD_{\test}^{(1-k)}\big)\\
\ge~&  g_{f,\rho}\big(g^{-1}_{f,\rho}(1-\alpha)\big) - 
g'_{f,\rho}\big(g^{-1}(1-\alpha)\big) 
\times \bigg\{\sup_{t \in \cT(\alpha)} 2\cdot\textnormal{EstErr}^{(k)}(t)
+ \sqrt{
\frac{16w_{\max}^2}{|\cI_\tr^{(k)}|}
+ \frac{2}{|\cI_{\test,j}^{(k)}|}}
\bigg\}, 
\$
where $g'_{f,\rho}$ is the left derivative of $g_{f,\rho}$, and 
$\cT(\alpha) = [\underline{q},\bar{q}]$ is a neighborhood around 
the $g_{f,\rho}^{-1}(1-\alpha)$-th quantile under $Q_X\times P_{Y\given X}$ with
\$ 
\underline{q} = \sup_{0 \le t \le q^*(0)} \esterr(t) + 
\sqrt{
\frac{9w_{\max}^2}{|\cI_\tr^{(k)}|}
+ \frac{1}{|\cI_{\test,j}^{(k)}|}}, 
\quad 
\bar{q} = q^*(0).
\$
When $g_{f,\rho}(1) \ge 1-\alpha$, we further have 
\$  
& \PP_{(X_{n+j},Y_{n+j})\sim Q_{X,Y}}\big(Y_{n+j} \in \hat{C}_{f,\rho,n+j}^{\textnormal{DR}}(X_{n+j})
\biggiven \cD^{(1-k)}_{\tr}, \cD_{\test}^{(1-k)}\big)\\
\ge~&  1-\alpha - 
g'_{f,\rho}\big(g^{-1}(1-\alpha)\big) 
\times \bigg\{\sup_{ 
t \in \cT(\alpha)} 2\cdot \textnormal{EstErr}^{(k)}(t)
+ \sqrt{
\frac{16w_{\max}^2}{|\cI_\tr^{(k)}|}
+ \frac{2}{|\cI_{\test,j}^{(k)}|}}
\bigg\}. 
\$
\end{theorem}
The proof of Theorem~\ref{thm:dr_pi} is deferred to Appendix~\ref{appx:proof_dr}, 
where we prove a more general result that the prediction set is 
valid with high probability conditional on the training data; 
we then show how the general result implies Theorem~\ref{thm:dr_pi}.

Theorem~\ref{thm:dr_pi} implies that the miscoverage rate of 
$\hat{C}^{\textnormal{DR}}(X_{n+j})$ is the product of the local 
estimation error
plus an $O(n^{-1/2})$ term, where the product term is 
small if either $\hat{w}$ is approximately {\em proportional} to $w$, or 
if $\hat{m}(x ;t)$ is close to $m(x;t)$ in the 
neighborhood of the $g^{-1}_{f,\rho}(1-\alpha)$-th quantile under 
$Q_X\times P_{Y\given X}$. 
Compared with the double robustness result of~\citet{yang2022doubly}, 
our dependence on the estimation error of $\hat{m}(x;t)$ is local 
(around the $g^{-1}(1-\alpha)$-th quantile) while that of~\citet{yang2022doubly}
is global (for all $t$). This is achieved through the ``monotonization'' step, 
an idea that also appears in~\citet{gui2023conformalized}.

\subsection{Choice of the robust parameter $\rho$}
Another important piece of our procedure is the 
robust parameter 
$\rho$. Choosing $\rho$ is a common challenge in the distributionally 
robust learning literature 
(see e.g.,~\citet{rahimian2019distributionally,cauchois2023robust,
si2023distributionally,mu2022factored} and the references therein).
In certain applications, users can specify an appropriate $\rho$
with context-dependent knowledge. When the choice of $\rho$ is
not clear a priori, we provide two solutions based on the 
proposal of~\citet{si2023distributionally}:
\begin{enumerate}
\item [(1)] If there is (a small amount of) supervised data from 
target distribution, i.e., $Q_{X,Y}$, one can estimate an upper bound 
of $\rho$, and use the estimator in place of $\rho$.
\item [(2)] When no supervised data in the target distribution is 
available, we can apply the procedure with a sequence of  
$\rho$, obtaining a sequence of prediction sets. Each value of $\rho$ 
corresponds to a certain level of robustness, and the user can trade off 
between the level of robustness and efficiency (e.g., the length of the 
prediction interval).
\end{enumerate}

In the case where $\rho$ is estimated, Theorem~\ref{thm:est_rho} 
characterizes the coverage guarantee of WRCP. We only present the 
result for WRCP here for simplicity; the result extends also to 
D-WRCP.

\begin{theorem}\label{thm:est_rho}
Under the same assumptions of Theorem~\ref{thm:cov_est}, 
suppose that $\hat{\rho}$ is 
independent of $(\cD_\tr,\cD_\test)$. 
Denote $\rho^* \,:=\, \textnormal{ess\,sup}_x ~D_f(Q_{Y\given X=x} 
\,\|\, P_{Y\given X=x})$. 
Then for $k \in \{0,1\}$ and any $n+j \in \cI_\test^{(k)}$, 
the prediction interval produced by Algorithm~\ref{alg:wrcp}  
with the robust parameter taken to be $\hat{ \rho}$
satisfies
\begin{align}
\mathbb{P}\big(Y_{n+j} \in \hat{C}_{f,\hat{\rho},n+j}(X_{n+j})\big)
& \ge g_{f,\rho^*}\big(g^{-1}_{f,\hat{\rho}}(1-\alpha)\big) 
- \frac{1}{2}g_{f,\rho^*}'\big(g^{-1}_{f,\hat{\rho}}(1-\alpha)\big)
\cdot \mathbb{E}_{X\sim P_X}\bigg[\Big|\frac{\hat{w}^{(k)}(X)}
{\EE[\hat{w}^{(k)}(X)]}-w(X)\Big|\bigg],
\end{align}
where $g'_{f,\rho}$ is the left derivative of $g_{f,\rho}$.
Furthermore, if $g_{f,\hat{\rho}}(1) \ge 1-\alpha$ and 
$\hat{\rho}\ge \rho^*$,
then 
\$ 
\mathbb{P}\big(Y_{n+j}\in \hat{C}_{f,\hat{\rho},n+j}(X_{n+j})\big)\geq 
1-\alpha
- \frac{1}{2}g_{f,\rho^*}'\big(g^{-1}_{f,\hat{\rho}}(1-\alpha)\big)
\cdot \mathbb{E}_{X\sim P_X}\bigg[\Big|\frac{\hat{w}^{(k)}(X)}
{\EE[\hat{w}^{(k)}(X)]}-w(X)\Big|\bigg],
\$
where $g_{f,\rho}'$ is the left derivative of $g_{f,\rho}$.
\end{theorem}

\section{Application: sensitivity analysis of individual treatment effects}\label{sec:4}
Our framework can be applied to the sensitivity analysis 
of individual treatment effects in the presence of confounding factors.
To set the stage, we follow the potential outcome framework~\citep{neyman1923applications,imbens2015causal}
and suppose that each sample is associated with a set of random variables  
$(X,U,T,Y(0),Y(1))$, where $X\in \cX$ denotes the observed covariates, 
$U \in \cU$ the unobserved confounders, $T \in \{0,1\}$ the binary treatment, and 
$Y(1),Y(0) \in \RR$ the potential outcomes with and without being treated.
Here, not all the quantities are observed --- the observable variables 
are $(X,T,Y)$, where the realized outcome $Y = T Y(1) + (1-T) Y(0)$ 
under the {\em Stable Unit Treatment Value Assumption (SUTVA)}.
Assume that the unobserved confounder $U$ satisfies that\footnote{
Such an assumption can always be achieved 
by taking $U$ to be $(Y(1),Y(0))$.
}
\$ 
(Y(1), Y(0)) \, \indep \, T \given X,U,
\$

Imagine now there is a cohort of $n$ i.i.d samples 
$(X_i,U_i,T_i,Y_i(1),Y_i(0))_{i=1}^n$, where  
we observe $\cD = (X_i,T_i,Y_i)_{i=1}^n$.
For a new individual $X_{n+1}$, we are interested in a prediction 
interval $\hat{C}(X_{n+1})$ 
for individual treatment effect (ITE), $Y(1) -  Y(0)$, such that  
\@\label{eq:ite_pi}
\PP\big(Y(1) -  Y(0) \in \hat{C}(X_{n+1})\big) \ge 1-\alpha.
\@

Without additional constraints on the unobserved confounders, 
it is hopeless to obtain an efficient prediction interval 
achieving~\eqref{eq:ite_pi}, since the difference in the treated and control group 
can be entirely driven by the 
confounding factor. Previously,~\citet{lei2021conformal} 
studies this problem assuming that there are no observed confounders, i.e., 
$(Y(1),Y(0))\, \indep\, T \given X$;~\citet{jin2023sensitivity} 
adopts the marginal $\Gamma$-selection 
model~\citep{tan2006distributional}, which allows for unobserved 
confounders but the influence of 
$U$ --- roughly speaking --- is {\em uniformly} bounded
by a constant $\Gamma$.
The marginal $\Gamma$-selection model can be unsatisfactory 
in some cases, where the influence of $U$ is limited only 
{\em on average} but is unbounded with small probability
(the corresponding constant $\Gamma$ is therefore $+\infty$).
Such a situation can however be well characterized by the 
$f$-sensitivity model~\citep{jin2022sensitivity}:
\begin{definition}[The $(f,\rho)$-selection condition]
Suppose $f: \RR_+ \mapsto \RR$ is  a convex function such that $f(1)=0$, and
$P$ is a distribution over $(X,U,T,Y(1),Y(0))$. $P$ satisfies the 
$(f,\rho)$-selection condition if for $P$-almost all $x$, 
\$
\int f\Big(\frac{e(X)}{1-e(X)} \frac{1 - \bar{e}(X,U)}{\bar{e}(X,U)}\Big) 
\,\textnormal{d} P_{U \given X = x,T=1} \le \rho, \text{ and } 
\int f\Big(\frac{1-e(X)}{e(X)} \frac{\bar{e}(X,U)}{1-\bar{e}(X,U)}\Big) 
\,\textnormal{d} P_{U \given X = x,T=0} \le \rho,
\$
where $\bar{e}(x,u) = P(T=1 \given X = x, U = u)$ and 
$e(x) = P(T=1 \given X = x)$.
\end{definition}

Can we construct a prediction interval achieving~\eqref{eq:ite_pi} under 
the $f$-sensitivity model? 
It turns out that this task is
a special case of our proposed framework. To see this, we first reduce the 
problem to that of inference on the counterfactuals: if we can construct  
valid prediction intervals for $Y(1)$ and $Y(0)$, respectively, then combining these two intervals and taking a union bound yields a valid interval for the ITE.

Without loss of generality, we focus on $Y(1)$, aiming to construct an interval 
$\hat{C}_{f,\rho}(X_{n+1})$ such that $\PP(Y(1) \in \hat{C}_{f,\rho}(X_{n+1})) \ge 
1-\alpha$. Since $Y(1)$ can only be observed for the treated units, the 
training data follows the distribution $P_{Y(1),X \given T=1}$ while 
our target distribution is $P_{Y(1),X}$ --- there exists a distributional shift.
The covariate shift can be computed  as follows
\$ 
w(x) = \frac{\dd P_{X}}{\dd P_{X \given T = 1}}(x) = \frac{\PP(T=1)}{e(x)} 
\propto \frac{1}{e(x)},
\$
which depends only on the observable propensity score and can be estimated with the data.
Next, we consider the distributional shift in $Y(1) \given X$. 
By~\citet[Lemma 1]{jin2022sensitivity}, under the $f$-sensitivity model,  
$D_f(P_{Y(1) \given X,T=0} \,\| \,P_{Y(1)\given X,T=1})\leq \rho$ almost surely. 
Consequently, 
\$ 
D_f\big(P_{Y(1) \given X} \,\|\, P_{Y(1) \given X,T=1}\big) 
& = D_f\big(e(X) \cdot P_{Y(1) \given X, T=1} + 
(1-e(X)) \cdot P_{Y(1)\given X,T=0} \,\|\, P_{Y(1) \given 
X,T=1}\big) \\
& \le (1-e(X))\cdot D_f(P_{Y(1) \given X,T=0} \,\|\, P_{Y(1)\given X,T=1}) \le \rho,
\$
where the inequality follows from the convexity of the $f$-divergence.
By now, it should be clear that the distributional shift in our task consists 
of an estimable covariate shift and a shift in $Y \given X$ bounded in $f$-divergence, 
and therefore fits into the framework of this paper. For completeness, we present the 
adaptation of our main proposal to this specific task of sensitivity analysis, as long as results for other types of estimands in Appendix~\ref{appx:ite}.



\section{Numerical results}
\label{sec:simulation}
In this section, we present several representative settings and leave 
the other results to Appendix~\ref{appx:simulation}.
\subsection{Simulation setup and evaluation metrics}
We empirically compare our proposed methods
\texttt{WRCP} and \texttt{D-WRCP} with the following benchmarks:
\begin{itemize}
\item [-] \texttt{CP}: standard conformal prediction 
designed for exchangeable~data~\citep{vovk2005algorithmic}; 
\item [-] \texttt{WCP}: weighted conformal prediction~\citep{tibshirani2019conformal};
\item [-] \texttt{RCP}: robust conformal prediction~\citep{cauchois2023robust}.
\end{itemize}
For all the five candidate methods, we implement the split version, where 
half of the data is reserved for model fitting and the other half for 
calibration. The nonconformity score $s(x,y) = |y - \hat{\mu}(x)|$ is adopted, 
where we fit $\hat{\mu}(\cdot)$ with cross-validated Lasso~\citep{tibshirani1996regression} 
using the 
\texttt{scikit-learn} package in python~\citep{pedregosa2011scikit}. 
For \texttt{WCP}, \texttt{WRCP} and \texttt{D-WRCP}, 
the covariate likelihood ratio 
$w(x)$ is estimated via the random 
forest classifier~\citep{breiman2001random} in the 
\texttt{scikit-learn} package.
For \texttt{WRCP}, \texttt{D-WRCP}, we use 
the KL divergence to quantify the distributional shift, 
i.e., $f(t) = t\log t$, and consider 
a sequence of robust parameters $\rho$. 
For each $\rho$, the 
corresponding robust parameter of \texttt{RCP}
is chosen as  
$\rho_{\text{RCP}} = \rho + D_{\text{KL}}(Q_X \,\|\, P_X)$ 
by the chain rule of KL divergence, where $D_{\text{KL}}(Q_X \,\|\, P_X)$
is estimated by plugging in the estimated $\hat{w}$.
In the implementation of \texttt{D-WRCP}, 
the conditional CDF is estimated by random forest with 
the python package \texttt{qosa-indices}~\citep{Elie-Dit-Cosaque2020}.
For all methods, the target coverage rate is $0.9$.

In our simulations, we consider $X \in \RR^{50}$ and $Y \in \RR $. 
For the training data,
\$
X \sim \cN(0,I_{50}), \quad 
Y \given X \sim X^\top\beta+\mathcal{N}(0,1)
\$
where $\|\beta\|_0 = 10$ and the nonzero entries take the value $0.47$.
The target covariate distribution has a shifted mean:
\$
Q_X = \cN(\beta_0, I_{50}), \text{ and } \beta_0=(\eta,-\eta,0,\cdots,0),
\$ 
where $\eta$ is a tuning parameter controlling the amount of 
covariate shift; the target $Y\given X$ distribution is specified 
as follows,
\begin{align*}
\frac{\dd Q_{Y \given X}}{\dd P_{Y \given X}}(x)&=
\begin{cases}
0.96 & \text{ if }\big|Y-X^\top\beta\big|<1.86;\\
1.59& \text{ if } \big|Y-X^\top\beta\big|\geq 1.86.
\end{cases}
\end{align*}
By construction, the ground truth 
$\rho^* = D_{\text{KL}}(Q_{Y\given X} \,\|\, P_{Y\given X}) = 0.01$.

We let $\eta$ to be $0.1$, $0.5$, and  $0.8$ --- corresponding to 
low, medium, and high levels of covariate shift respectively.
For each run of under a simulation setting, a training set $\cD_\tr$ and 
a test set $\cD_\test$  are generated,
with $|\cD_\tr| = |\cD_{\test}| = 2000$.
We consider $\rho \in  \{0.005,0.01,\ldots,0.025\}$. For 
each $\rho$, the above experiment is repeated for $N = 100$ runs, and 
for each method, 
we compute the averaged coverage rate and 
prediction interval length
averaged over the $100$ runs and 
50\% of the test samples ($1000$ samples):
\$ 
\widehat{\text{Coverage}} = \frac{1}{100\times 1000}\sum^{100}_{i=1} 
\sum_{j=1}^{1000} \ind\{Y_{ij} \text{ covered}\}, \qquad
\widehat{\text{Length}} = \frac{1}{100 \times 1000}\sum^{100}_{i=1}
\sum^{1000}_{j=1}
\text{Length}_{ij}.
\$
Ideally, a method should have $\widehat{\text{Coverage}} \ge 0.9$
and as small $\widehat{\text{Length}}$ as possible.


\subsection{Simulation results}
Figure~\ref{fig:sim-res} presents the simulation results of all methods.
As expected, \texttt{CP} and \texttt{WCP} fail to achieve the desired coverage
level $0.9$.
\texttt{RCP} is overly conservative since it also considers the worst-case covariate 
shift. Our proposed method \texttt{WRCP} and \texttt{D-WRCP} 
achieve approximate validity for a wide range of 
$\rho$ --- in particular, \texttt{WRCP} and \texttt{D-WRCP} achieve 
almost exact coverage when $\rho = \rho^*$;  
as $\rho$ increases, the coverage remains reasonably close 
to the target level.

The prediction interval length tells a similar story: 
\texttt{CP} and \texttt{WCP} have short prediction intervals 
due to undercoverage; \texttt{RCP} often outputs 
prediction intervals of infinite length (for the purpose 
of illustration, we replace $+\infty$ with $17$ --- an upper bound 
of all the realized lengths --- when plotting the results); 
our methods provide valid and informative prediction intervals.

\begin{figure}[h]
    \centering
    \includegraphics[width = 0.8\textwidth]{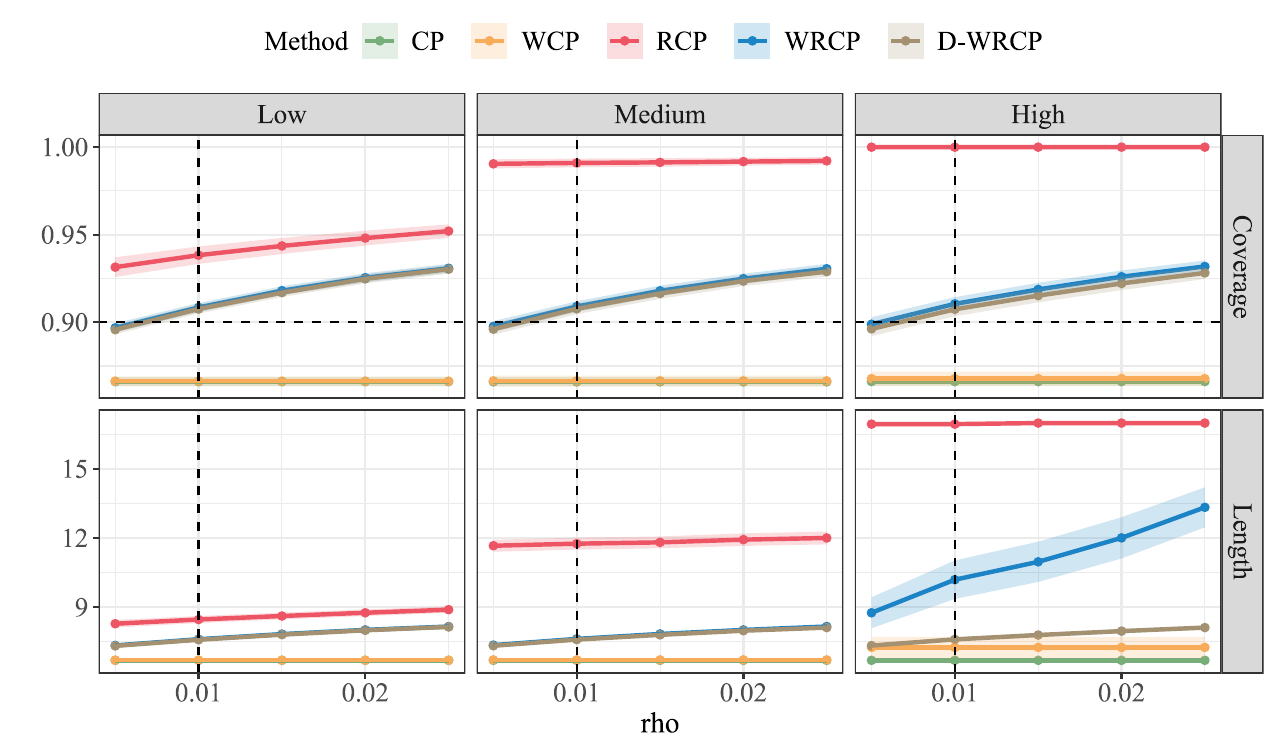}
    \caption{Averaged coverage (top) and prediction interval length (bottom)
    over $N=100$ independent runs as a function of the robust parameter $\rho$, when the amount of covariate shift  is 
    low (left), medium (middle), and high (right).
    The shaded bars correspond to the $95\%$ confidence intervals. 
    The horizontal dashed line corresponds to the target coverage rate $0.9$, 
    and the vertical dashed line is the true robust parameter $\rho^* = 0.01$.
    }
    \label{fig:sim-res}
\end{figure}

\section{Real data applications}
In this section, we evaluate the performance of 
all methods on four real datasets:
the national study of learning mindsets dataset~\citep{carvalho2019assessing}, 
the ACS income dataset~\citep{ding2021retiring}, and 
the covid information study datasets~\citep{pennycook2020fighting,roozenbeek2021accurate}, 
and the poverty mapping dataset~\citep{yeh2020using,wilds2021}. 
The results corresponding to the 
first three datasets are presented in the main text, while the results for the 
last dataset are deferred to Appendix~\ref{appx:real_data}.
For each task, we implement \texttt{WRCP} and \texttt{D-WRCP}, as well 
as the benchmarks \texttt{CP}, \texttt{WCP}, and \texttt{RCP} with the 
target coverage level $80\%$. For \texttt{WRCP} and \texttt{D-WRCP}, 
we choose $f(t) = t\log t$ (the KL-divergence),
and consider a sequence of robust parameter $\rho$'s, 
reporting the averaged coverage rate and prediction set length/cardinality
as a function of $\rho$. For \texttt{RCP}, the robust parameter is 
chosen as $\rho_{\text{RCP}} = \rho + D_{\rm KL}(Q_X\,\|\,P_X)$, where 
$D_{\rm KL}(Q_X \, \| \,P_X) = \EE_{Q_X}[\log(\dd Q_X/\dd P_X)]$ 
is estimated with Monte Carlo with the set-aside training data.

\subsection{National study of learning mindsets}
The  National Study of Learning Mindsets (NSLM)~\citep{yeager2019national,yeager2019national2}
is a randomized study investigating the effect of instilling students with a growth mindset.
Based on the results from NSLM,~\citet{carvalho2019assessing} creates 
an observational study dataset with similar characteristics to the original study. 
The observational study dataset contains $10,\!391$ students, where 
$3,\!384$ received the intervention and $7,\!007$ did not.
For each student, the dataset records 
the treatment status $T$, the outcome $Y$, and 
ten covariates: S3, C1, C2, C3, XC, X1, X2, X3, X4, 
X5.\footnote{The detailed description of the covariates can be found 
in~\citet[Table 1]{carvalho2019assessing}. The original 
dataset also contains the school id, but we do not include
it in our analysis.}
We consider the task of predicting $Y(1)$ in the 
control group, where we wish to construct 
a prediction interval $\hat{C}_{f,\rho}(X)$ such that 
\$ 
\PP\big(Y(1) \in \hat{C}_{f,\rho}(X) \biggiven T = 0\big)\ge 80\%.
\$

Without observing the counterfactuals, 
the validity of a procedure cannot be evaluated.
As an alternative, we create a semi-synthetic data
based on the NSLM dataset. Following the strategy 
of~\citet{carvalho2019assessing},
we generate synthetic potential outcomes from the 
following model:
\$
Y(t)=\mu(x)+\tau(x_1,x_2,c_1) \cdot t+\epsilon, \text{ for }t=0,1.
\$
Above, $x$ denotes all the covariates for a student, and 
$x_1,x_2,c_1$ corresponds to X1, X2 and C1, respectively.
The baseline function 
$\mu$ is obtained by fitting a generalized additive model~\citep{hastie2017generalized} 
on the control arm of the original data, 
and $\epsilon$ is sampled with replacement from 
the sum of the residual from the fitted model on the original data and 
a noise term $\cN(0,0.025)$.
The form of the treatment effect is:
\$
\tau(x_1,x_2,c_1)=0.228 + 0.05\cdot\ind\{x_1 < 0.07\}
- 0.05 \cdot\mathds{1}\{x_2< -0.69\}
-0.08\cdot \mathds{1}\{c_1\in \{1,13,14\}\}.
\$
Confounding is introduced by removing two features, X1 and X2.

We consider a sequence of $\rho \in \{0.001,0.005,0.01,0.015,\ldots,0.04\}$.  
For each run, we randomly select half of the treated units and half of the 
control units for model fitting;
the other half of the treated units are reserved for calibration, while  
the other half of the control units for evaluation.
The nonconformity score function
$s(x,y) = |y - \hat{\mu}(x)|$.
Both the regression function $\hat{\mu}(x)$ 
and the propensity score function $e(x)$ 
are fitted with random forest.
With each $\rho$, we repeat the above process for $100$ random splits.

Figure~\ref{fig:nslm} plots the resulting coverage rate and prediction 
interval length as a function of $\rho$.  
\texttt{CP} and \texttt{WCP} fail to achieve the desired coverage level, 
while \texttt{RCP} is overly conservative, achieving a 
much higher coverage rate than the target level.
Our methods \texttt{WRCP} and \texttt{D-WRCP}
achieve approximate coverage for a wide range of $\rho$, and are much 
more efficient than RCP.

\begin{figure}[htbp]
    \centering
    \includegraphics[width = 0.8\textwidth]{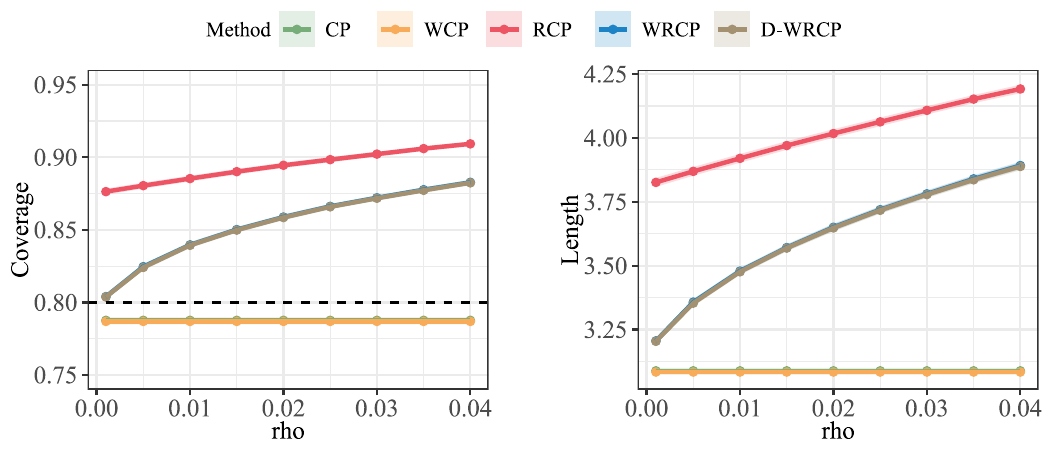}

    \caption{Averaged coverage (left) and prediction interval length (right)
    over $100$ runs as a function of the robust parameter $\rho$ from 
    experiments on the NSLM dataset. The dashed bar corresponds to the 
    $95\%$ confidence interval and  
    the horizontal dashed line corresponds to the target coverage rate $80\%$.}
    \label{fig:nslm}
\end{figure}

\subsection{ACS income dataset}
We further evaluate our procedure for a classification task with 
the ACS income dataset constructed from the US census data~\citep{ding2021retiring}.
The target is to predict whether an individual's annual income is above 
$50,\!000$ dollars. 
The specific dataset we use is the version pre-processed by~\citet{liu2023need},
where we choose the data from New York (NY) as the 
training set, and that from South Dakota (SD) as the target ---
as discussed in~\citet{liu2023need}, both $X$-shift and $Y\given X$-shift 
exist between the training and target population. 
The training  and test set contain $103,\!021$ and 
$4,\!899$ samples, respectively. For each sample, $9$ features are recorded, 
where $3$ of them are continuous and $6$ categorical. After introducing the 
dummy variables, the dimension of $X$ comes to $76$. The response variable 
$Y = \ind\{\text{income} \ge 50,\!000 
\text{ dollars}\}$.

Since $Y$ is binary, we adopt the generalized inverse quantile
conformity score introduced by~\citet{romano2020classification},
and the prediction set is a subset of $\{0,1\}$. 
The weight function $\hat{w}$ is estimated via 
XGBoost~\citep{chen2016xgboost} 
with the hyperparameters provided by~\citet{liu2023need},
and the outcome model $\hat{\mu}$ is fitted with random forest.
The robust parameter $\rho \in \{0.005,0.01,\ldots,0.04\}$.
In each run, we randomly select $2000$ samples from the source 
population, and the same number of samples from the target population. 
For each $\rho$, we repeat the above process over $100$ random splits and report the averaged
coverage and prediction set cardinality.

Figure~\ref{fig:acs} presents the simulation results of all methods.
Ae before, \texttt{CP} and \texttt{WCP} fail to achieve the desired 
coverage level $80\%$; \texttt{WCP} improves upon \texttt{CP} 
because it adjusts for the covariate shift.
\texttt{RCP} is overly conservative, while our methods again 
deliver valid and efficient 
prediction intervals for a wide range of $\rho$'s.
\begin{figure}[htbp]
    \centering
    \includegraphics[width = 0.8\textwidth]{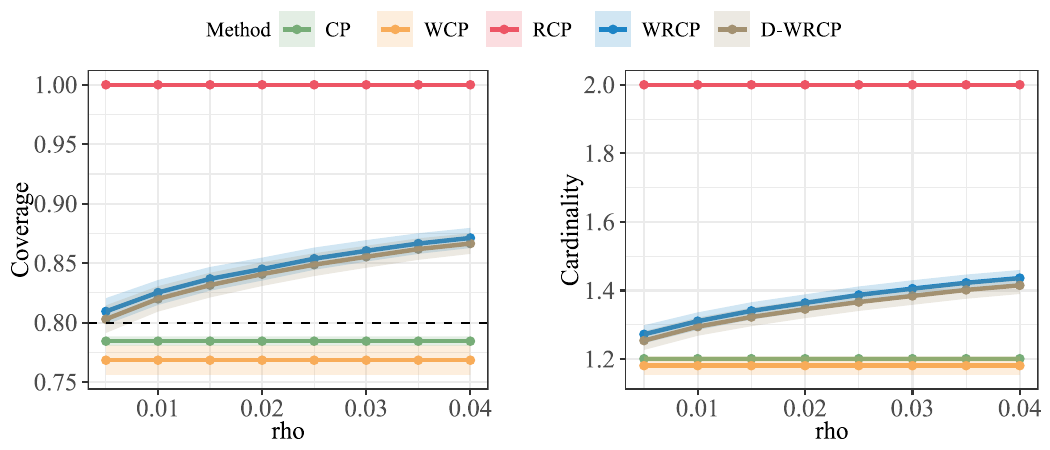}
    \caption{Averaged coverage (left) and prediction set cardinality (right)
    over $100$ runs as a function of the robust parameter $\rho$
    from the experiment on ACS income dataset. The other details are the same as in Figure~\ref{fig:nslm}.} 
    \label{fig:acs}
\end{figure}

\subsection{COVID information studies}
The covid information studies investigate how a ``nudge'' for 
thinking about the accuracy of information can affect the people's
ability to discern fake news when sharing COVID-related headlines. The original 
study of $856$ participants~\citep{pennycook2020fighting} is first conducted, 
followed by a replication study~\citep{roozenbeek2021accurate} of $1,\!583$
participants. The original study found a significant interaction term between 
the intervention and the validity of the headline, while the replication study 
also found a significant interaction, but with a much smaller magnitude (more 
details about the comparison between the two studies can be found in~\citet{jin2023diagnosing}).

As discussed in~\citet{jin2023diagnosing}, the discrepancy between the two 
results can be attributed to the distributional shift in the both 
the covariates and $Y\given X$. Here, instead of estimating the treatment effect,
we consider the task of predicting a participant's rating for willingness to share a headline.
Each sample in the dataset corresponds to a participant, where the outcome is 
the rating for their willingness to share a headline; the predictors include the 
treatment status (i.e., whether a nugde is sent), the validity of the news, and 
$10$ other covariates.\footnote{In the original datasets, each participant was 
asked to rate $30$ headlines. In our analysis, the outcome and the covariates 
are all averaged over the $30$ headlines.}
After removing the samples with missing values, the training and test set 
consist of $811$ and $1,\!583$ samples, respectively.
Each run splits the 
training and test sets into two halves for model fitting and calibration;
The weight function $\hat{w}$ and the outcome model $\hat{\mu}$
are both estimated via random forest. 
The robust parameter $\rho \in \{0.005,0.01,\ldots,0.04\}$, and for each $\rho$
We repeat the above process for $100$ random splits.

Figure~\ref{fig:covid} demonstrates the results of all methods. 
For the purpose of visualization, we replace the infinite prediction interval 
length with $2$ (an upper bound of all the finite realized lengths) when plotting 
the averaged prediction interval length. In this example, we again see 
that \texttt{CP} and \texttt{WCP} fail to achieve the desired coverage level,
with \texttt{WCP} being slightly better than \texttt{CP} due to the adjustment for 
covariate shift. The proposed methods \texttt{WRCP} and \texttt{D-WRCP} achieve 
approximate coverage for a wide range of $\rho$'s, and are much more 
efficient than \texttt{RCP}. 

\begin{figure}[ht]
    \centering
    \includegraphics[width = 0.8\textwidth]{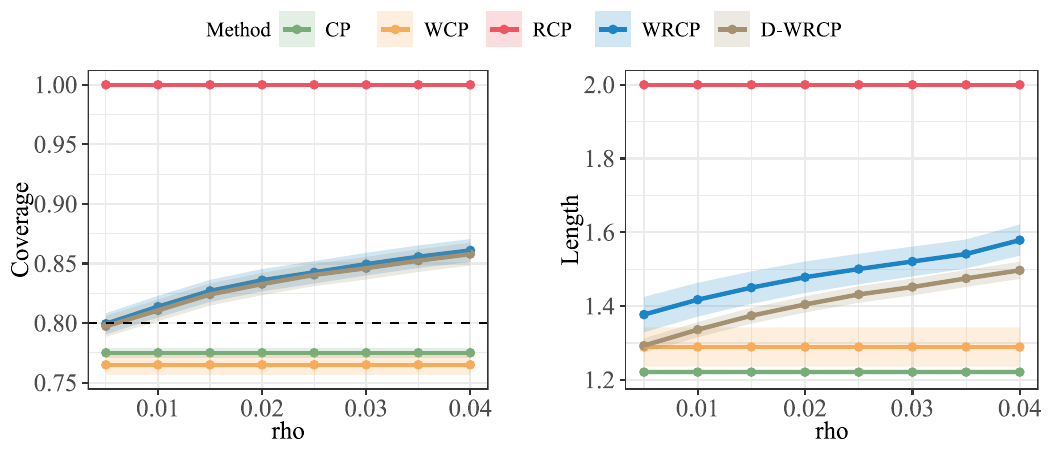}
    \caption{Averaged coverage (left) and prediction interval length (right)
    over $100$ runs as a function of the robust parameter $\rho$
    from the experiment on covid study datasets. 
    The other details are the same as in Figure~\ref{fig:nslm}.}
    \label{fig:covid}
\end{figure}
\section{Discussion}
In this paper, we provide a fine-grained approach to 
quantifying the uncertainty of predictive models by 
distinguishing sources of distributional shift and providing 
different treatments. We propose two new methods, WRCP and its 
debiased version D-WRCP, that achieve 
validity and efficiency under a wide range of distributional shifts, 
as demonstrated in the simulation and real data experiments.
This paper opens up several interesting directions for future work. First, 
it would be interesting to investigate methods for identifying (an upper bound) 
of the robust parameter $\rho$ when we have a small amount of supervised data 
from the target population. Second, can we extend this fine-grained approach 
to other distributional shift models and improves the efficiency of the 
corresponding methodologies? Last but not least, there can be other ways to 
decompose the distributional shifts --- it remains to be understood
the optimal decomposition in different settings and the corresponding 
treatments. 

\subsection*{Reproducibility}
All the numerical results in this paper can be reproduced with the code
available at \url{https://github.com/zhimeir/finegrained-conformal-paper}.

\subsection*{Acknowledgements}
The authors would like to thank Wharton High Performance Computing for 
the computational resources and the support from the staff members.
The authors would also like to thank Ying Jin for the feedback on 
the manuscript.

\bibliographystyle{apalike}
\bibliography{ref}

\begin{thebibliography}{}

\bibitem[Angelopoulos et~al., 2023]{angelopoulos2023conformal}
Angelopoulos, A.~N., Bates, S., et~al. (2023).
\newblock Conformal prediction: A gentle introduction.
\newblock {\em Foundations and Trends{\textregistered} in Machine Learning}, 16(4):494--591.

\bibitem[Barber et~al., 2021]{foygel2021limits}
Barber, R.~F., Candes, E.~J., Ramdas, A., and Tibshirani, R.~J. (2021).
\newblock The limits of distribution-free conditional predictive inference.
\newblock {\em Information and Inference: A Journal of the IMA}, 10(2):455--482.

\bibitem[Barber et~al., 2023]{barber2023conformal}
Barber, R.~F., Candes, E.~J., Ramdas, A., and Tibshirani, R.~J. (2023).
\newblock Conformal prediction beyond exchangeability.
\newblock {\em The Annals of Statistics}, 51(2):816--845.

\bibitem[Blanchet and Murthy, 2019]{blanchet2019quantifying}
Blanchet, J. and Murthy, K. (2019).
\newblock Quantifying distributional model risk via optimal transport.
\newblock {\em Mathematics of Operations Research}, 44(2):565--600.

\bibitem[Breiman, 2001]{breiman2001random}
Breiman, L. (2001).
\newblock Random forests.
\newblock {\em Machine learning}, 45:5--32.

\bibitem[Carvalho et~al., 2019]{carvalho2019assessing}
Carvalho, C., Feller, A., Murray, J., Woody, S., and Yeager, D. (2019).
\newblock Assessing treatment effect variation in observational studies: Results from a data challenge.

\bibitem[Cauchois et~al., 2023]{cauchois2023robust}
Cauchois, M., Gupta, S., Ali, A., and Duchi, J.~C. (2023).
\newblock Robust validation: Confident predictions even when distributions shift.
\newblock {\em Journal of the American Statistical Association}, (just-accepted):1--22.

\bibitem[Chen and Guestrin, 2016]{chen2016xgboost}
Chen, T. and Guestrin, C. (2016).
\newblock Xgboost: A scalable tree boosting system.
\newblock In {\em Proceedings of the 22nd acm sigkdd international conference on knowledge discovery and data mining}, pages 785--794.

\bibitem[Chernozhukov et~al., 2021]{chernozhukov2021distributional}
Chernozhukov, V., W{\"u}thrich, K., and Zhu, Y. (2021).
\newblock Distributional conformal prediction.
\newblock {\em Proceedings of the National Academy of Sciences}, 118(48):e2107794118.

\bibitem[Cover, 1999]{cover1999elements}
Cover, T.~M. (1999).
\newblock {\em Elements of information theory}.
\newblock John Wiley \& Sons.

\bibitem[Ding et~al., 2021]{ding2021retiring}
Ding, F., Hardt, M., Miller, J., and Schmidt, L. (2021).
\newblock Retiring adult: New datasets for fair machine learning.
\newblock {\em Advances in neural information processing systems}, 34:6478--6490.

\bibitem[Duchi et~al., 2023]{duchi2023distributionally}
Duchi, J., Hashimoto, T., and Namkoong, H. (2023).
\newblock Distributionally robust losses for latent covariate mixtures.
\newblock {\em Operations Research}, 71(2):649--664.

\bibitem[Duchi and Namkoong, 2021]{duchi2021learning}
Duchi, J.~C. and Namkoong, H. (2021).
\newblock Learning models with uniform performance via distributionally robust optimization.
\newblock {\em The Annals of Statistics}, 49(3):1378--1406.

\bibitem[Elie-Dit-Cosaque, 2020]{Elie-Dit-Cosaque2020}
Elie-Dit-Cosaque, K. (2020).
\newblock qosa-indices.

\bibitem[Gendler et~al., 2021]{gendler2021adversarially}
Gendler, A., Weng, T.-W., Daniel, L., and Romano, Y. (2021).
\newblock Adversarially robust conformal prediction.
\newblock In {\em International Conference on Learning Representations}.

\bibitem[Ghosh et~al., 2023]{ghosh2023probabilistically}
Ghosh, S., Shi, Y., Belkhouja, T., Yan, Y., Doppa, J., and Jones, B. (2023).
\newblock Probabilistically robust conformal prediction.
\newblock In {\em Uncertainty in Artificial Intelligence}, pages 681--690. PMLR.

\bibitem[Guan, 2023]{guan2023localized}
Guan, L. (2023).
\newblock Localized conformal prediction: A generalized inference framework for conformal prediction.
\newblock {\em Biometrika}, 110(1):33--50.

\bibitem[Gui et~al., 2023]{gui2023conformalized}
Gui, Y., Hore, R., Ren, Z., and Barber, R.~F. (2023).
\newblock {Conformalized survival analysis with adaptive cutoffs}.
\newblock {\em Biometrika}, page asad076.

\bibitem[Gupta et~al., 2022]{gupta2022nested}
Gupta, C., Kuchibhotla, A.~K., and Ramdas, A. (2022).
\newblock Nested conformal prediction and quantile out-of-bag ensemble methods.
\newblock {\em Pattern Recognition}, 127:108496.

\bibitem[Hastie, 2017]{hastie2017generalized}
Hastie, T.~J. (2017).
\newblock Generalized additive models.
\newblock In {\em Statistical models in S}, pages 249--307. Routledge.

\bibitem[Imbens and Rubin, 2015]{imbens2015causal}
Imbens, G.~W. and Rubin, D.~B. (2015).
\newblock {\em Causal inference in statistics, social, and biomedical sciences}.
\newblock Cambridge University Press.

\bibitem[Jin et~al., 2023a]{jin2023diagnosing}
Jin, Y., Guo, K., and Rothenh{\"a}usler, D. (2023a).
\newblock Diagnosing the role of observable distribution shift in scientific replications.
\newblock {\em arXiv preprint arXiv:2309.01056}.

\bibitem[Jin et~al., 2023b]{jin2023sensitivity}
Jin, Y., Ren, Z., and Cand{\`e}s, E.~J. (2023b).
\newblock Sensitivity analysis of individual treatment effects: A robust conformal inference approach.
\newblock {\em Proceedings of the National Academy of Sciences}, 120(6):e2214889120.

\bibitem[Jin et~al., 2022]{jin2022sensitivity}
Jin, Y., Ren, Z., and Zhou, Z. (2022).
\newblock Sensitivity analysis under the $ f $-sensitivity models: a distributional robustness perspective.
\newblock {\em arXiv preprint arXiv:2203.04373}.

\bibitem[Kallus and Zhou, 2020]{kallus2020confounding}
Kallus, N. and Zhou, A. (2020).
\newblock Confounding-robust policy evaluation in infinite-horizon reinforcement learning.
\newblock {\em Advances in neural information processing systems}, 33:22293--22304.

\bibitem[Kallus and Zhou, 2021]{kallus2021minimax}
Kallus, N. and Zhou, A. (2021).
\newblock Minimax-optimal policy learning under unobserved confounding.
\newblock {\em Management Science}, 67(5):2870--2890.

\bibitem[Koh et~al., 2021]{wilds2021}
Koh, P.~W., Sagawa, S., Marklund, H., Xie, S.~M., Zhang, M., Balsubramani, A., Hu, W., Yasunaga, M., Phillips, R.~L., Gao, I., Lee, T., David, E., Stavness, I., Guo, W., Earnshaw, B.~A., Haque, I.~S., Beery, S., Leskovec, J., Kundaje, A., Pierson, E., Levine, S., Finn, C., and Liang, P. (2021).
\newblock {WILDS}: A benchmark of in-the-wild distribution shifts.
\newblock In {\em International Conference on Machine Learning (ICML)}.

\bibitem[Lei et~al., 2018]{lei2018distribution}
Lei, J., G’Sell, M., Rinaldo, A., Tibshirani, R.~J., and Wasserman, L. (2018).
\newblock Distribution-free predictive inference for regression.
\newblock {\em Journal of the American Statistical Association}, 113(523):1094--1111.

\bibitem[Lei and Cand{\`e}s, 2021]{lei2021conformal}
Lei, L. and Cand{\`e}s, E.~J. (2021).
\newblock Conformal inference of counterfactuals and individual treatment effects.
\newblock {\em Journal of the Royal Statistical Society Series B: Statistical Methodology}, 83(5):911--938.

\bibitem[Lei et~al., 2023]{lei2023policy}
Lei, L., Sahoo, R., and Wager, S. (2023).
\newblock Policy learning under biased sample selection.
\newblock {\em arXiv preprint arXiv:2304.11735}.

\bibitem[Liu et~al., 2023]{liu2023need}
Liu, J., Wang, T., Cui, P., and Namkoong, H. (2023).
\newblock On the need for a language describing distribution shifts: Illustrations on tabular datasets.
\newblock {\em arXiv preprint arXiv:2307.05284}.

\bibitem[Miller et~al., 2020]{miller2020effect}
Miller, J., Krauth, K., Recht, B., and Schmidt, L. (2020).
\newblock The effect of natural distribution shift on question answering models.
\newblock In {\em International conference on machine learning}, pages 6905--6916. PMLR.

\bibitem[Mu et~al., 2022]{mu2022factored}
Mu, T., Chandak, Y., Hashimoto, T.~B., and Brunskill, E. (2022).
\newblock Factored {DRO}: Factored distributionally robust policies for contextual bandits.
\newblock {\em Advances in Neural Information Processing Systems}, 35:8318--8331.

\bibitem[Namkoong et~al., 2023]{namkoong2023diagnosing}
Namkoong, H., Yadlowsky, S., et~al. (2023).
\newblock Diagnosing model performance under distribution shift.
\newblock {\em arXiv preprint arXiv:2303.02011}.

\bibitem[Neyman, 1923]{neyman1923applications}
Neyman, J. (1923).
\newblock Sur les applications de la th{\'e}orie des probabilit{\'e}s aux experiences agricoles: Essai des principes.
\newblock {\em Roczniki Nauk Rolniczych}, 10(1):1--51.

\bibitem[Papadopoulos et~al., 2002]{papadopoulos2002inductive}
Papadopoulos, H., Proedrou, K., Vovk, V., and Gammerman, A. (2002).
\newblock Inductive confidence machines for regression.
\newblock In {\em Machine Learning: ECML 2002: 13th European Conference on Machine Learning Helsinki, Finland, August 19--23, 2002 Proceedings 13}, pages 345--356. Springer.

\bibitem[Park et~al., 2022]{park2022pac}
Park, S., Dobriban, E., Lee, I., and Bastani, O. (2022).
\newblock {PAC} prediction sets under covariate shift.
\newblock In {\em International Conference on Learning Representations}.

\bibitem[Pedregosa et~al., 2011]{pedregosa2011scikit}
Pedregosa, F., Varoquaux, G., Gramfort, A., Michel, V., Thirion, B., Grisel, O., Blondel, M., Prettenhofer, P., Weiss, R., Dubourg, V., et~al. (2011).
\newblock Scikit-learn: Machine learning in python.
\newblock {\em Journal of machine learning research}, 12(Oct):2825--2830.

\bibitem[Pennycook et~al., 2020]{pennycook2020fighting}
Pennycook, G., McPhetres, J., Zhang, Y., Lu, J.~G., and Rand, D.~G. (2020).
\newblock Fighting covid-19 misinformation on social media: Experimental evidence for a scalable accuracy-nudge intervention.
\newblock {\em Psychological science}, 31(7):770--780.

\bibitem[Podkopaev and Ramdas, 2021]{podkopaev2021distribution}
Podkopaev, A. and Ramdas, A. (2021).
\newblock Distribution-free uncertainty quantification for classification under label shift.
\newblock In {\em Uncertainty in Artificial Intelligence}, pages 844--853. PMLR.

\bibitem[Qiu et~al., 2022]{qiu2022distribution}
Qiu, H., Dobriban, E., and Tchetgen, E.~T. (2022).
\newblock Distribution-free prediction sets adaptive to unknown covariate shift.
\newblock {\em arXiv preprint arXiv:2203.06126}.

\bibitem[Rahimian and Mehrotra, 2019]{rahimian2019distributionally}
Rahimian, H. and Mehrotra, S. (2019).
\newblock Distributionally robust optimization: A review.
\newblock {\em arXiv preprint arXiv:1908.05659}.

\bibitem[Recht et~al., 2019]{recht2019imagenet}
Recht, B., Roelofs, R., Schmidt, L., and Shankar, V. (2019).
\newblock Do imagenet classifiers generalize to imagenet?
\newblock In {\em International conference on machine learning}, pages 5389--5400. PMLR.

\bibitem[Romano et~al., 2019]{romano2019conformalized}
Romano, Y., Patterson, E., and Candes, E. (2019).
\newblock Conformalized quantile regression.
\newblock {\em Advances in neural information processing systems}, 32.

\bibitem[Romano et~al., 2020]{romano2020classification}
Romano, Y., Sesia, M., and Candès, E.~J. (2020).
\newblock Classification with valid and adaptive coverage.

\bibitem[Roozenbeek et~al., 2021]{roozenbeek2021accurate}
Roozenbeek, J., Freeman, A.~L., and van~der Linden, S. (2021).
\newblock How accurate are accuracy-nudge interventions? a preregistered direct replication of pennycook et al.(2020).
\newblock {\em Psychological science}, 32(7):1169--1178.

\bibitem[Rosenbaum, 1987]{rosenbaum1987sensitivity}
Rosenbaum, P.~R. (1987).
\newblock Sensitivity analysis for certain permutation inferences in matched observational studies.
\newblock {\em Biometrika}, 74(1):13--26.

\bibitem[Sahoo et~al., 2022]{sahoo2022learning}
Sahoo, R., Lei, L., and Wager, S. (2022).
\newblock Learning from a biased sample.
\newblock {\em arXiv preprint arXiv:2209.01754}.

\bibitem[Shafieezadeh~Abadeh et~al., 2015]{shafieezadeh2015distributionally}
Shafieezadeh~Abadeh, S., Mohajerin~Esfahani, P.~M., and Kuhn, D. (2015).
\newblock Distributionally robust logistic regression.
\newblock {\em Advances in Neural Information Processing Systems}, 28.

\bibitem[Si et~al., 2023]{si2023distributionally}
Si, N., Zhang, F., Zhou, Z., and Blanchet, J. (2023).
\newblock Distributionally robust batch contextual bandits.
\newblock {\em Management Science}.

\bibitem[Tan, 2006]{tan2006distributional}
Tan, Z. (2006).
\newblock A distributional approach for causal inference using propensity scores.
\newblock {\em Journal of the American Statistical Association}, 101(476):1619--1637.

\bibitem[Tibshirani, 1996]{tibshirani1996regression}
Tibshirani, R. (1996).
\newblock Regression shrinkage and selection via the lasso.
\newblock {\em Journal of the Royal Statistical Society Series B: Statistical Methodology}, 58(1):267--288.

\bibitem[Tibshirani et~al., 2019]{tibshirani2019conformal}
Tibshirani, R.~J., Foygel~Barber, R., Candes, E., and Ramdas, A. (2019).
\newblock Conformal prediction under covariate shift.
\newblock {\em Advances in neural information processing systems}, 32.

\bibitem[Vovk et~al., 2005]{vovk2005algorithmic}
Vovk, V., Gammerman, A., and Shafer, G. (2005).
\newblock {\em Algorithmic learning in a random world}, volume~29.
\newblock Springer.

\bibitem[Wong et~al., 2021]{wong2021external}
Wong, A., Otles, E., Donnelly, J.~P., Krumm, A., McCullough, J., DeTroyer-Cooley, O., Pestrue, J., Phillips, M., Konye, J., Penoza, C., et~al. (2021).
\newblock External validation of a widely implemented proprietary sepsis prediction model in hospitalized patients.
\newblock {\em JAMA Internal Medicine}, 181(8):1065--1070.

\bibitem[Yadlowsky et~al., 2018]{yadlowsky2018bounds}
Yadlowsky, S., Namkoong, H., Basu, S., Duchi, J., and Tian, L. (2018).
\newblock Bounds on the conditional and average treatment effect with unobserved confounding factors.
\newblock {\em arXiv preprint arXiv:1808.09521}.

\bibitem[Yang et~al., 2022]{yang2022doubly}
Yang, Y., Kuchibhotla, A.~K., and Tchetgen, E.~T. (2022).
\newblock Doubly robust calibration of prediction sets under covariate shift.
\newblock {\em arXiv preprint arXiv:2203.01761}.

\bibitem[Yeager, 2019]{yeager2019national2}
Yeager, D.~S. (2019).
\newblock {\em The National Study of Learning Mindsets,[United States], 2015-2016}.

\bibitem[Yeager et~al., 2019]{yeager2019national}
Yeager, D.~S., Hanselman, P., Walton, G.~M., Murray, J.~S., Crosnoe, R., Muller, C., Tipton, E., Schneider, B., Hulleman, C.~S., Hinojosa, C.~P., et~al. (2019).
\newblock A national experiment reveals where a growth mindset improves achievement.
\newblock {\em Nature}, 573(7774):364--369.

\bibitem[Yeh et~al., 2020]{yeh2020using}
Yeh, C., Perez, A., Driscoll, A., Azzari, G., Tang, Z., Lobell, D., Ermon, S., and Burke, M. (2020).
\newblock Using publicly available satellite imagery and deep learning to understand economic well-being in africa.
\newblock {\em Nature Communications}.

\bibitem[Yin et~al., 2022]{yin2022conformal}
Yin, M., Shi, C., Wang, Y., and Blei, D.~M. (2022).
\newblock Conformal sensitivity analysis for individual treatment effects.
\newblock {\em Journal of the American Statistical Association}, pages 1--14.

\bibitem[Zhang et~al., 2023]{zhang2023optimal}
Zhang, Z., Zhan, W., Chen, Y., Du, S.~S., and Lee, J.~D. (2023).
\newblock Optimal multi-distribution learning.
\newblock {\em arXiv preprint arXiv:2312.05134}.

\bibitem[Zhao et~al., 2019]{zhao2019sensitivity}
Zhao, Q., Small, D.~S., and Bhattacharya, B.~B. (2019).
\newblock Sensitivity analysis for inverse probability weighting estimators via the percentile bootstrap.
\newblock {\em Journal of the Royal Statistical Society Series B: Statistical Methodology}, 81(4):735--761.

\end{thebibliography}

\newpage
\appendix
\section{Auxiliary lemmas}
\begin{lemma}[Data processing inequality~\citep{cover1999elements}]
\label{lem:dpi}
Let $X,Y,Z$ denote random variables
drawn from a Markov chain in the order 
(denoted by $X \rightarrow Y \rightarrow Z$) that 
the conditional distribution of $Z$
depends only on $Y$ and is conditionally independent of $X$. 
Then if $X \rightarrow Y \rightarrow Z$, we have 
$I(X; Y) \ge I(X;Z)$,
where $I(X;Y)$ is the mutual information between $X$ and $Y$.
\end{lemma}

\begin{lemma}[Adapted from Lemma A.1 of~\citet{cauchois2023robust}]
\label{lem:property}
Let $f: \RR \mapsto \RR$ be a closed convex function such that 
$f(1) = 0$ and $f(t) < \infty$ for all $t > 0$. 
The function $g_{f,\rho}(\beta)$ has the following properties.
\begin{itemize}
\item [(a)] $(\rho,\beta)\rightarrow g_{f,\rho}(\beta)$ is a 
convex function and continuous in $\beta \in [0,1]$ and $\rho \in (0,\infty)$.
\item [(b)] $g_{f,\rho}(\beta)$ is non-increasing in $\rho$ and 
non-decreasing in $\beta$. Moreover, for all $\rho > 0$, 
there exists $\beta_0(\rho) \,:=\, \sup \big\{\beta\in (0,1) \given 
g_{f,\rho}(\beta)=0 \big\}$, and $g_{f,\rho}(\beta)$ is 
strictly increasing for $\beta > \beta_0(\rho)$.
\end{itemize}
\end{lemma}

\section{Technical Proofs}
\subsection{Proof of Theorem \ref{thm:cov_known_shift}}
\label{appx:proof_cov_known_shift}
For notional simplicity, let $A \,:=\, 
\big\{Y_{n+1}\in \hat{C}_{f,\rho}(X_{n+1})\big\}$.
We aim to show that
\$
\PP_{(X_{n+1},Y_{n+1})\sim Q_{X,Y}}(A) \ge 1-\alpha.
\$
Denote by $\bern(p)$ the 
Bernoulli distribution with success probability $p$.
By data processing inequality (Lemma~\ref{lem:dpi}), 
we have that
\begin{align}
D_f\big(Q_{Y\given X} \,\|\, 
P_{Y \given X}\big) \ge
D_f\big(\bern\big(\mathbb{P}_{Y_{n+1}\sim Q_{Y \given X}}(A \given X_{n+1},\cD)\big)
\,\big\|\, \bern\big(\mathbb{P}_{Y_{n+1}\sim P_{Y\given X}}(A \given X_{n+1},\cD)\big)\big).
\end{align}
Recall that $Q \in \cP(\rho;P)$. We then have 
\begin{align}
\rho & \geq D_f\big(Q_{Y \given X}\,\|\,P_{Y \given X}\big)\\
      & \geq D_f\big(\bern\big(\mathbb{P}_{Y_{n+1} \sim 
      Q_{Y \given X}}(A \given X_{n+1}, \cD)\big)
      \, \big\|\, \bern\big(\mathbb{P}_{Y_{n+1} 
      \sim P_{Y \given X}}(A \given X_{n+1}, \cD)\big)\big)\\
      & = \mathbb{P}_{Y_{n+1}\sim P_{Y \given X}}(A \given X_{n+1},\cD)\cdot 
      f\bigg(\frac{\mathbb{P}_{Y_{n+1}\sim Q_{Y\given X}}(A \given X_{n+1}, \cD)}
      {\mathbb{P}_{Y_{n+1}\sim P_{Y \given X}}(A\given X_{n+1}, \cD)}\bigg)\\
      &\qquad \qquad \qquad \qquad \qquad 
      +\big(1-\mathbb{P}_{Y\sim P_{Y \given X}}(A \given X_{n+1}, \cD)\big)
      \cdot 
      f\bigg(\frac{1-\mathbb{P}_{Y\sim Q_{Y\given X}}(A \given X_{n+1}, \cD)}
      {1-\mathbb{P}_{Y\sim P_{Y \given X}}(A \given X_{n+1},\cD)}\bigg),
\end{align}
where the last step follows from the definition of the 
$f$-divergence.
Combining the above and the definition of $g_{f,\rho}$, 
one can obtain that almost surely
\begin{align}
\label{eq:quant_ub}
  g_{f,\rho}\big(\mathbb{P}_{Y\sim P_{Y \given X}}(A)\big)
  \leq \mathbb{P}_{Y\sim Q_{Y \given X}}(A).
\end{align}
Next, we take the expectation over the randomness of 
the training set $\cD$ and $X_{n+1}$
\begin{align}
\label{eq:step1}
\mathbb{P}_{\cD,\, 
(X_{n+1},Y_{n+1})\sim Q_{X,Y}}(A)
&=\mathbb{E}_{\cD, \, 
X_{n+1}\sim Q_X}\big[
\mathbb{P}_{Y_{n+1}\sim Q_{Y \given X}}
\big(A \given X_{n+1},\cD\big)\big] \nonumber\\
& \geq\mathbb{E}_{\cD,
X_{n+1}\sim Q_X}
\Big[ g_{f,\rho}\big(\mathbb{P}_{Y\sim P_{Y \given X}}(A \given X_{n+1},
\cD) \big)\Big],
\end{align}
where the last inequality is because of~\eqref{eq:quant_ub}.
Since $g_{f,\rho}$ is convex (Lemma~\ref{lem:property}), 
Jensen's inequality implies
\begin{align}
\label{eq:step2}
\mathbb{E}_{\cD,X_{n+1}\sim Q_X}\Big[
g_{f,\rho}\big(\mathbb{P}_{Y\sim P_{Y \given X}}(A \given X_{n+1},\cD)\big)\Big]
& \ge g_{f,\rho}\Big(
\mathbb{P}_{\cD,
(X_{n+1},Y_{n+1}) \sim Q_{X} \times P_{Y\given X}}(A)\Big).
\end{align}
By~\citet[Corollary 1]{tibshirani2019conformal} and the 
construction of $\hat{C}_{f,\rho}(X_{n+1})$, there is  
\begin{align}
  \mathbb{P}_{\cD,(X_{n+1},Y_{n+1})\sim Q_{X}\times P_{Y \given X}}(A) 
  \geq g_{f,\rho}^{-1}\big(1-\alpha\big).
\end{align}
Again leveraging the monotonicity of $g_{f,\rho}(\beta)$ in $\beta$, we 
arrive at
\begin{align}
\mathbb{P}_{\cD, 
(X_{n+1},Y_{n+1})\sim Q_{X,Y}}(A)
&\geq g_{f,\rho}\big(g_{f,\rho}^{-1}(1-\alpha)\big).
\end{align}
When $g_{f,\rho}(1) \ge 1-\alpha$, we prove that 
$g_{f,\rho}\big(g^{-1}_{f,\rho} (1-\alpha)\big) \ge 1-\alpha$
by contradiction. Suppose otherwise that 
$g_{f,\rho}\big(g^{-1}_{f,\rho} (1-\alpha)\big) < 1-\alpha$.
Then $g^{-1}_{f,\rho}(1-\alpha)<1$. By the continuity of 
$g_{f,\rho}(\beta)$ in $\beta$, there exists a small 
$\varepsilon >0$, such that $g_{f,\rho}(\beta) < 1-\alpha$ for all 
$\beta \in [g^{-1}_{f,\rho}(1-\alpha), g^{-1}_{f,\rho}(1-\alpha)+\epsilon]$,
which contradicts the definition of $g^{-1}_{f,\rho}(1-\alpha)$.
The proof is therefore completed.

\subsection{Proof of Theorem \ref{thm:cov_est}}
\label{appx:proof_cov_est}
Throughout the proof, we condition on $\cD_\tr^{(0)}\cup \cD_\test^{(1-k)}$,
and for notational simplicity we do not explicitly write the conditional 
probability and expectation when the context is clear. 

We start by defining a new distribution $\widetilde{Q}_{X,Y}=
\widetilde{Q}_X \times P_{Y \given X}$ 
such that 
\$
\frac{d\tilde{Q}_X}{dP_X}(x) = 
\frac{\hat{w}^{(k)}(x)}{\EE_{X \sim P_X}[\hat{w}^{(k)}(X)]}
.
\$ 
Let $A = \big\{Y_{n+j} \in \hat{C}_{f,\rho,j}(X_{n+j})\big\}$.
If $(X_{n+j},Y_{n+j})$ were indeed sampled from $\widetilde{Q}_{X,Y}$,
then by Lemma 1 of~\citet{tibshirani2019conformal}, we have 
\begin{align}
  \label{eq:tildeq}
  \mathbb{P}_{\cD_{\tr}^{(1)},
  (X_{n+j},Y_{n+j})\sim \widetilde{Q}_{X,Y}}(A)\geq g_{f,\rho}^{-1}(1-\alpha).
\end{align}
Meanwhile, by the definition of the TV distance, 
\begin{align}
&\Big|\mathbb{P}_{(X_{n+j},Y_{n+j})\sim \widetilde{Q}_X\times P_{Y \given X}}
\big(A \given \cD_\tr^{(1)}\big) -
\mathbb{P}_{(X_{n+j},Y_{n+j})\sim Q_X \times P_{Y\given X}}
\big(A \given \cD_\tr^{(1)}\big)\Big|\\
\leq ~&D_{\text{TV}}\big(\widetilde{Q}_X\times Q_{Y\given X}\,\|\, 
Q_X\times Q_{Y \given X}\big)\\
=~& D_{\text{TV}}(\widetilde{Q}_X \,\|\, Q_X)
= \frac{1}{2}\mathbb{E}_{X\sim P_X}
\bigg[\Big|\frac{\hat{w}^{(k)}(X)}{\EE_{X\sim P_X}[\hat{w}^{(k)}(X)]}-w(X)\Big|\bigg].
\end{align}
The above inequality implies that
\begin{align}
\mathbb{P}_{(X_{n+j},Y_{n+j})\sim Q_X\times P_{Y\given X}}
\big(A \biggiven \cD_\tr^{(1)} \big) \geq
\mathbb{P}_{(X_{n+j},Y_{n+j})
\sim \widetilde{Q}_X \times P_{Y \given X}}\big(A \biggiven \cD_\tr^{(1)}\big)
-\frac{1}{2}\mathbb{E}_{X\sim P_X}
\bigg[\Big|\frac{\hat{w}^{(k)}(X)}{\EE_{X\sim P_X}[\hat{w}^{(k)}(X)]}-w(X)\Big|\bigg].
\end{align}
Taking expectation over $\cD_\tr^{(1)}$, we have
\begin{align}
& \PP_{\cD_\tr^{(1)},
(X_{n+j},Y_{n+j})\sim {Q}_X \times P_{Y\given X}} (A)\\
\geq~& \PP_{\cD_\tr^{(1)},
(X_{n+j},Y_{n+j})\sim \widetilde{Q}_{X,Y}}(A)
-\frac{1}{2}\mathbb{E}_{X\sim P_X}
\bigg[\Big|\frac{\hat{w}^{(k)}(X)}{\EE_{X\sim P_X}[\hat{w}^{(k)}(X)]}-w(X)\Big|\bigg]\\
\geq~& g_{f,\rho}^{-1}(1-\alpha)
-\frac{1}{2}\mathbb{E}_{X\sim P_X}
\bigg[\Big|\frac{\hat{w}^{(k)}(X)}{\EE_{X\sim P_X}[\hat{w}^{(k)}(X)]}-w(X)\Big|\bigg],
\end{align}
where the last step follows from~\eqref{eq:tildeq}.
Following the same argument as in the proof of Theorem~\ref{thm:cov_known_shift} 
(Eqn.~\eqref{eq:step1} and~\eqref{eq:step2}), we have
\begin{align}
\label{eq:first_result}
& \mathbb{P}_{\cD_\tr^{(1)},
(X_{n+1},Y_{n+1})\sim Q_{X,Y}}(A) \nonumber \\
\geq ~& g_{f,\rho}\Big(\mathbb{P}_{\cD_\tr^{(1)},
(X_{n+j},Y_{n+j})\sim Q_X\times P_{Y \given X}}(A)\Big)\nonumber\\
\geq ~& g_{f,\rho}\Bigg(g_{f,\rho}^{-1}(1-\alpha)-
\frac{1}{2}\mathbb{E}_{X\sim P_X}\bigg[\Big|\frac{\hat{w}^{(k)}(X)}
{\EE_{X\sim P_X}[\hat{w}^{(k)}(X)]}-w(X)\Big|\bigg]\Bigg),
\end{align}
where the last inequality is because $g_{f,\rho}(\beta)$ is non-decreasing in 
$\beta$. Since $g_{f,\rho}(\beta)$ is convex in $(0,1)$, the left derivative 
exists and by the separating hyperplane theorem, we further have 
\$
\eqref{eq:first_result} 
\ge g_{f,\rho}\big(g^{-1}_{f,\rho}(1-\alpha)\big) - 
\frac{1}{2}g'_{f,\rho}\big(g^{-1}_{f,\rho}(1-\alpha)\big)
\cdot \mathbb{E}_{X\sim P_X}\bigg[\Big|\frac{\hat{w}^{(k)}(X)}
{\EE_{X\sim P_X}[\hat{w}^{(k)}(X)]}-w(X)\Big|\bigg].
\$
As shown in the proof of Theorem~\ref{thm:cov_known_shift}, 
when $g_{f,\rho}(1) \ge 1-\alpha$, 
$g_{f,\rho}\big(g_{f,\rho}^{-1}(1-\alpha)\big) \ge 1-\alpha$, and 
we conclude the proof.

\subsection{Proof for the doubly robust prediction intervals}
We start by proving that the $D_{f,\rho,n+j}^{\text{DR}}(X_{n+j})$ 
is training-conditionally valid in Theorem~\ref{thm:hp_dr_pi}, and 
then show that Theorem~\ref{thm:dr_pi}  
is a direct consequence of Theorem~\ref{thm:hp_dr_pi}.
\label{appx:proof_dr}
\begin{theorem}
\label{thm:hp_dr_pi}
For any $k \in \{0,1\}$, assume that 
\begin{enumerate}
\item [(1)]  $\hat{w}^{(k)}(x) \le w_{\max} \cdot \EE_{P_X}[\hat{w}^{(k)}(X)]$;
\item [(2)] $\hat{m}^{(k)}(x;t) \in [0,1]$ is non-decreasing and right-continuous in $t$.
\end{enumerate}
Denote the product estimation error by 
\$ 
\textnormal{EstErr}^{(k)}(t) =
\big\|m^{(k)}(X;t) - \hat{m}^{(k)}(X;t)\big\|_{L_2(P)}
\bigg\|\frac{\hat{w}^{(k)}(X)}{\EE_{X \sim P_X}[\hat{w}^{(k)}(X)]}
- w(X)\bigg\|_{L_2(P)}, 
\$
where $m^{(k)}(X;t) := \PP_{Y \given X \sim P_{Y\given X}}(s(X,Y;\hat \mu^{(k)})\le t \given X=x)$, 
 $\|\cdot \|_{L_2(P)}$ denotes the $L_2$-norm under 
$P$, and the expectation is taken conditional on $\cD_\tr^{(1-k)}$
and $\cD_\test^{(1-k)}$.
Then for any $\delta>0$ and any unit $n+j \in \cI_\test^{(k)}$, 
with probability at least $1-\delta$, 
\$  
& \PP_{(X_{n+j},Y_{n+j})\sim Q_{X,Y}}
\big(Y_{n+j} \in \hat{C}_{f,\rho,n+j}^{\textnormal{DR}}(X_{n+j})\biggiven 
\cD_\tr^{(k)}, \cD_{\test,j}^{(k)}, \cD_\tr^{(1-k)}, \cD_{\test}^{(1-k)}\big)\\
\ge~&  g_{f,\rho}\big(g^{-1}_{f,\rho}(1-\alpha)\big) - 
g'_{f,\rho}\big(g^{-1}(1-\alpha)\big) 
\times \Bigg\{\sup_{t \in \cT(\alpha)}
\textnormal{EstErr}^{(k)}(t)
+ \sqrt{\log\Big(\frac{1}{\delta}\Big)
\cdot \bigg(\frac{9w_{\max}^2}{|\cI_\tr^{(k)}|}
+ \frac{1}{|\cI_{\test,j}^{(k)}|}\bigg)}
\Bigg\}, 
\$
where $g'_{f,\rho}$ is the left derivative of $g_{f,\rho}$. 
If $g_{f,\rho}(1) \ge 1-\alpha$,  then with probability at least $1-\delta$, 
\$  
& \PP_{(X_{n+j},Y_{n+j})\sim Q_{X,Y}}\big(Y_{n+j} \in \hat{C}_{f,\rho,n+j}^{\textnormal{DR}}(X_{n+j})\biggiven 
\cD_\tr^{(k)}, \cD_{\test,j}^{(k)}, \cD_{\tr}^{(1-k)}, \cD_{\test}^{(1-k)}\big)\\
\ge~&  1-\alpha - 
g'_{f,\rho}\big(g^{-1}(1-\alpha)\big) 
\times \Bigg\{\sup_{t \in \cT(\alpha)} \textnormal{EstErr}^{(k)}(t)
+ \sqrt{\log\bigg(\frac{1}{\delta}\Big)\cdot
\Big(\frac{9w_{\max}^2}{|\cI_\tr^{(k)}|}
+ \frac{1}{|\cI_{\test,j}^{(k)}|}\bigg)}
\Bigg\}. 
\$
\end{theorem}

\begin{proof}
Without loss of generality, assume $k = 1$.  
Throughout, we condition on 
$\cD_\tr^{(0)} \cup \cD_\test^{(0)}$ without explicitly writing 
the conditioning event when the context is clear.
For notational simplicity, we define the normalized weight as 
\$ 
\tw^{(k)}(x) = \frac{\hat{w}^{(k)}(x)}{\EE_{X\sim P_X}[\hat{w}^{(k)}(X)]}.
\$
In the proof we leave out the 
dependence on $k$, writing $m(X;t)$, $\hat{m}(X;t)$ and $\tw(X)$  
in place of $m(X;t)$, $\hat{m}^{(1)}(X;t)$ and $\tw^{(1)}(X)$; additionally, we refer to 
$\EE_{(X,Y) \sim P_{X,Y}}$ as $\EE_{P_{X,Y}}$, 
with the same rule applied to the expectation/probability 
under other distributions, and let $n_\tr = |\cI_\tr^{(1)}|$
and $n_\test = |\cI_{\test,j}^{(1)}|$. 

For any $t \in \RR$, we define the oracle CDF 
$F(t) = \PP_{Q_X\times P_{Y \given X}}\big(s(X,Y; \hat \mu^{(1)}) \le t\big)$, 
and for any $\xi \in [0,g^{-1}_{f,\rho}(1-\alpha)]$, the 
perturbed oracle quantile $q^*(\xi)$ can be equivalently written as
\$ 
q^*(\xi) = \inf\big\{t \in \RR: F(t) \ge g^{-1}_{f,\rho}(1-\alpha) - \xi\big\}.
\$
Consider the error of margin
\$
& \Delta = \sup_{\kappa \cdot q^*(\bar{\Delta})\le  t \le q^*(0)} 
\esterr(t)
 + \sqrt{\log\Big(\frac{1}{\delta}\Big)\cdot 
\Big(\frac{9w_{\max}^2}{n_\tr} + \frac{1}{n_\test}\Big) },
\$ 
where $\kappa \in (0,1)$ is a constant that can be arbitrarily close to $1$
and 
\$
&\bar{\Delta} = \Bigg\{\sup_{0 \le  t \le q^*(0)} 
\esterr(t)
 + \sqrt{\log\Big(\frac{1}{\delta}\Big)
\cdot 
\Big(\frac{9w_{\max}^2}{n_\tr} + \frac{1}{n_\test}\Big) } \Bigg\}
\wedge g^{-1}_{f,\rho}(1-\alpha).
\$
Here, $a \wedge b = \min(a,b)$ and $\Delta$ is fully deterministic conditional 
on $\cD_\tr^{(0)} \cup \cD_\test^{(0)}$.

The proof consists of two steps: (1) we show that,
with high probability, $\hat{q} \ge q^*(\Delta)$.
and therefore $\hat{q}$ is no less than the 
$(g^{-1}_{f,\rho}(1-\alpha) - \Delta)$-th quantile
under $Q_X\times P_{Y \given X}$, and (2) $\hat{q}$ is 
approximately
an upper bound of the $(1-\alpha)$-th quantile under $Q_{X,Y}$.

\paragraph{Step (1).}
On the event $\{q^*(\Delta) \le \hat{q}\}$, 
\$ 
\PP_{Q_X\times P_{Y\given X}}\big(s(X_{n+j}, Y_{n+j}) \le \hat{q} \given 
\cD_\tr^{(1)}, \cD_{\test,j}^{(1)}\big) = F(\hat{q}) 
\stackrel{\text{(i)}}{\ge} F(q^*(\Delta)) 
\stackrel{\text{(ii)}}{\ge} g^{-1}_{f,\rho}(1-\alpha) - \Delta,
\$
where step (i) follows from the monotonicity of 
$F(t)$, and step (ii) is by the definition of $q^*(\Delta)$ 
and that $F(t)$ is right-continuous. It then suffices 
to control the probability of $\{q^*(\Delta) > \hat{q}\}$.
Fixing $\delta \in [0,1]$, we aim at showing that 
\$ 
\PP\big(\hat{q} < q^*(\Delta)\big) \le \delta. 
\$
The above is trivial when $\Delta \ge g^{-1}_{f,\rho}(1-\alpha)$. 
We proceed assuming that $\Delta < g^{-1}_{f,\rho}(1-\alpha)$.

For any $\varepsilon > 0$, we have that 
$F(q^*(\Delta) - \eps) < g^{-1}_{f,\rho}(1-\alpha) - \Delta$ 
by the definition of $q^*(\Delta)$.
If $\hat{q}\le q^*(\Delta) - \eps$, then 
the choice of $\hat{q}$ implies that 
\$ 
\hat{p}(q^*(\Delta) - \eps) \ge g^{-1}_{f,\rho}(1-\alpha).
\$
In other words, 
\begin{align} 
\label{eq:bad_event}
\PP\big(\hat{q} \le  q^*(\Delta) - \eps\big) 
\le ~& 
\PP\Big(\hat{p}(q^*(\Delta) - \eps) \ge g^{-1}_{f,\rho}(1-\alpha)\Big)
.
\end{align}
For better readability, we use $\Bar{t}$ to represent 
$q^*(\Delta) - \eps$, and let 
$Z_i = \ind\{S_i \le \bart\} - \hat{m}(X_i;\bart)$
in the following. The right-hand side of~\eqref{eq:bad_event} can be 
further upper bounded as
\@\label{eq:dr_step1} 
& \PP\big(\hat{p}(\bart) \ge g^{-1}_{f,\rho}(1-\alpha)\big)\nonumber\\
=~& \PP\bigg(\sum_{i \in \cI_\tr^{(1)}} \tw(X_i)Z_i + 
\Big(\sum_{i \in \cI_\tr^{(1)}} \tw(X_i)\Big)
\Big(\frac{1}{n_\test}
\sum_{i \in \cI_{\test,j}^{(1)}} \hat{m}(X_i;\bart)
- g_{f,\rho}^{-1}(1-\alpha)\Big) \ge 0 \bigg) \nonumber\\
\le~& \EE\Bigg[\exp\bigg\{\eta\bigg(\sum_{i \in \cI_\tr^{(1)}} \tw(X_i)Z_i + 
\Big(\sum_{i \in \cI_\tr^{(1)}} \tw(X_i)\Big)
\Big(\frac{1}{n_\test}\sum_{i \in \cI_{\test,j}^{(1)}} \hat{m}(X_i;\bart)
- g_{f,\rho}^{-1}(1-\alpha)\Big)  
\bigg)\bigg\}\Bigg],
\@
where $\eta > 0$ is some constant to be determined and 
the last step follows from Markov's inequality.

Conditional on $\{X_i\}_{i\in \cI_\tr^{(1)}}$, 
$Z_i - \EE[Z_i \given X_i]$ are $\frac{1}{4}$-subgaussian random variables.
Therefore,
\$ 
\EE\bigg[\exp\Big\{\eta\Big(\sum_{i \in \cI_\tr^{(1)}}
\tw(X_i) \cdot \big(Z_i-\EE[Z_i \given X_i]\big)\Big) \Big\} \Biggiven 
\{X_i\}_{i \in \cI_\tr^{(1)}}\bigg] 
\le \exp\Big(\frac{\eta^2 w_{\max}^2 n_\tr}{8}\Big). 
\$
The above implies that 
\$ 
\eqref{eq:dr_step1} \le ~& 
\exp\Big(\frac{\eta^2 w_{\max}^2 n_\tr}{8}\Big)
\cdot \EE\Bigg[\exp\Bigg\{
\eta\bigg(\sum_{i \in \cI_\tr^{(1)}} 
\tw(X_i)\EE[Z_i \given X_i] \\  
&\qquad\qquad\qquad\qquad\qquad \qquad  
+\Big(\sum_{i \in \cI_\tr^{(1)}} \tw(X_i)\Big)
\Big(\frac{1}{n_\test}
\sum_{i \in \cI_{\test,j}} \hat{m}(X_i;\bart)
- g_{f,\rho}^{-1}(1-\alpha)\Big)  
\bigg)\Bigg\}\Bigg]\\
\le~& \exp\Big(\frac{\eta^2 w_{\max}^2 n_\tr}{8}\Big)\cdot 
\EE\Bigg[\exp\Bigg\{
2\eta\bigg(\sum_{i \in \cI_\tr^{(1)}} 
\Big(\tw(X_i)\EE[Z_i \given X_i]  
+
\frac{1}{n_\test}\sum_{\ell \in \cI_{\test,j}^{(1)}} \hat{m}(X_\ell;\bart)
- g_{f,\rho}^{-1}(1-\alpha)\Big)
\bigg)\Bigg\}\Bigg]^{1/2}\\
& \qquad \qquad \times \EE\Bigg[\exp\bigg(2\eta
\sum_{i \in \cI_\tr^{(1)}} \big(\tw(X_i)-1\big)
\cdot \Big(\frac{1}{n_\test}
\sum_{\ell \in \cI_{\test,j}^{(1)}} \hat{m}(X_\ell;\bart)
- g_{f,\rho}^{-1}(1-\alpha)\Big)  \bigg)\Bigg]^{1/2}, 
\$
where the last inequality follows from the Cauchy-Schwarz inequality.
Since $\cD_\test^{(1)}$ is independent of $\cD_\tr^{(1)}$, conditional 
on $\cD_{\test}^{(1)}$, there is 
\$ 
\EE\Bigg[\exp\bigg(2\eta
\sum_{i \in \cI_\tr^{(1)}} \big(\tw(X_i)-1\big)
\Big(\frac{1}{n_\test}
\sum_{\ell \in \cI_{\test,j}^{(1)}} \hat{m}(X_\ell;\bart)
- g_{f,\rho}^{-1}(1-\alpha)\Big)  \bigg)\bigggiven \cD_{\test}^{(1)}\Bigg]^{1/2}
\le \exp\big(\eta^2w_{\max}^2 n_\tr\big),
\$
where we use the sub-gaussianity of $\tw(X_i)$.
Recalling that $F(\bart) < g^{-1}_{f,\rho}(1-\alpha) - \Delta$,
we have that 
\@ 
\label{eq:dr_step2}
&\EE\Bigg[\exp\Bigg\{2\eta\bigg(\sum_{i \in \cI_\tr^{(1)}} 
\Big(\tw(X_i)\EE[Z_i \given X_i]  
+
\frac{1}{n_\test}
\sum_{\ell \in \cI_{\test,j}^{(1)}} \hat{m}(X_\ell;\bart)
- g_{f,\rho}^{-1}(1-\alpha)\Big)
\bigg)\Bigg\}\Bigg]^{1/2}\nonumber\\
\le~&
\EE\Bigg[\exp\Bigg\{
2\eta\bigg(\sum_{i \in \cI_\tr^{(1)}} 
\Big(\tw(X_i)\EE[Z_i \given X_i]  
+\frac{1}{n_\test}
\sum_{\ell \in \cI_{\test,j}^{(1)}} \hat{m}(X_\ell;\bart)
- F(\bart) - \Delta \Big)
\bigg)\Bigg\}\Bigg]^{1/2}.
\@
By definition, 
$F(\bart) = \PP_{Q_X \times P_{Y \given X}}(s(X,Y) \le \bart)
= \EE_{P_{X}}[w(X)(m(X;\bart) - \hat m(X;\bart))] + \EE_{Q_X}[\hat{m}(X;\bart)]$. Therefore, 
\@
\label{eq:step3}
& \big|F(\bart) - \EE_{P_{X}}
[\tw(X)\EE_P[Z\given X]] - \EE_{Q_X}[\hat{m}(X;\bart)]\big|\nonumber \\
\le~& \bigg|\EE_{P_{X,Y}}
\Big[\big(w(X) - \tw(X)\big)
\cdot \big(m(X;\bart) - \hat{m}(X;\bart) \big)
\Big]\bigg|\nonumber \\
\le ~&\big\|w(X) - \widetilde{w}(X)\big\|_{L_2(P)} \cdot  
\big\|m(X;\bart) - \hat{m}(X;\bart)\big\|_{L_2(P)}
\nonumber\\
= ~& \esterr(\bart).
\@
The last inequality follows from the Cauchy-Schwarz inequality.

Next, we focus on the following quantity:
\@ \label{eq:step4}
&\EE\Bigg[\exp\bigg\{ 
2\eta\bigg(\sum_{i \in \cI_\tr^{(1)}} 
\Big(\tw(X_i)\EE[Z_i \given X_i] -\EE_{P_{X,Y}}[\tw(X)Z]\Big) 
+\frac{n_\tr}{n_\test}\sum_{i \in \cI_{\test,j}^{(1)}} 
\big(\hat{m}(X_i;\bart) - \EE_{Q_X}[\hat{m}(X;\bart)]\big)
\Big)
\bigg)\bigg\}\Bigg]\nonumber\\
=~&
\EE\Bigg[\exp\bigg\{ 
2\eta\sum_{i \in \cI_\tr^{(1)}} 
\Big(\tw(X_i)\EE[Z_i \given X_i] -\EE_{P_{X,Y}}[\tw(X)Z]\Big)
\bigg\}\Bigg]\nonumber\\
&\qquad \qquad \qquad \qquad 
+ \EE\Bigg[\exp\bigg\{
\frac{2\eta\cdot  n_\tr}{n_\test}\sum_{i \in \cI_{\test,j}^{(1)}} 
\Big(\hat{m}(X_i;\bart) - \EE_{Q_X}[\hat{m}(X;\bart)]\Big)
\bigg\}\Bigg]\nonumber\\
\le ~& \exp\Big(2\eta^2w_{\max}^2 n_\tr+ 
\frac{\eta^2 n_\tr^2}{2n_\test}\Big), 
\@
where the second step uses the independence between $\cD_\test$ and $\cD_\tr$. 
Combining~\eqref{eq:dr_step2},~\eqref{eq:step3}
and~\eqref{eq:step4} leads to 
\$ 
\eqref{eq:dr_step2} \le 
\exp\Big\{-\eta n_\tr\big(\Delta - \esterr(\bart)\big) 
+ \eta^2 \Big( w_{\max}^2 n_\tr + \frac{n_\tr^2}
{4n_\test} \Big) \Big\}.
\$
Putting everything together, we conclude that 
\@\label{eq:step5}
& \PP\Big(\hat{p}(q^*(\Delta) - \eps) \ge g^{-1}_{f,\rho}(1-\alpha)\Big)
\le 
\exp\Big\{-\eta n_\tr \big(\Delta - \esterr(\bart)\big) 
+ \eta^2 \Big( \frac{17  w_{\max}^2 n_\tr}{8} + \frac{n_\tr^2}
{4n_\test} \Big) \Big\}.
\@
We choose $\eta$ to minimize the upper bound above and get
\$
\eta = \frac{\Delta -\esterr(\bart)}
{\frac{17w_{\max}^2}{4} + \frac{n_\tr}{2n_\test}},
\text{ and correspondingly, }
\eqref{eq:step5} 
\le \exp\bigg\{ - 
\frac{\big(\Delta -\esterr(\bart)\Big)^2}
{\frac{9w_{\max}^2}{n_\tr} + \frac{1}{n_\test}}\bigg\}.
\$
Recall that $\bart = q^*(\Delta) - \eps$. Since $\Delta \le \bar{\Delta}$,
$\bart \ge 
q^*(\bar{\Delta}) - \eps$.
For $\eps$ sufficiently small, we further have
$q^*(\Delta) - \eps > q^*(\bar{\Delta}) \cdot \kappa$. 
By the definition of $\Delta$,  we have 
\$ 
\esterr(\bart) + \sqrt{\log\Big(\frac{1}{\delta}\Big)\cdot 
\Big(\frac{9w_{\max}^2 }{n_\tr} + \frac{1}{n_\test}\Big)}
\le \sup_{\kappa \cdot q^*(\bar{\Delta}) \le t \le q^*(0)}
\esterr(t) + \sqrt{\log\Big(\frac{1}{\delta}\Big)\cdot 
\Big(\frac{9w_{\max}^2 }{n_\tr} + \frac{1}{n_\test}\Big)}
= \Delta.
\$
Consequently, we arrive at 
$ 
\PP\big(\hat{q} \le q^*(\Delta) - \eps\big) \le \delta.
$ Taking $\eps \rightarrow 0$ and by the continuity  of 
the probability measure, we have that 
$\PP\big(\hat{q} < q^*(\Delta)\big)\le \delta $.

\paragraph{Step (II).}
Let $A = \{S_{n+j} \le q^*(\Delta)\}$.
As in the proof of Theorem~\ref{thm:cov_known_shift}, 
we have that 
\$ 
\PP_{Q_{X,Y}}\big(S_{n+j} \le q^*(\Delta) \big)\ge 
g_{f,\rho}\big(\PP_{Q_X\times P_{Y\given X}}
\big(S_{n+j}\le q^*(\Delta)\big)\big)  
\ge g_{f,\rho}\big(g^{-1}_{f,\rho}(1-\alpha) - \Delta\big).
\$
On the event $\{\hat{q} \ge q^*(\Delta)\}$,
\$
\PP_{Q_{X,Y}}\big(S_{n+j} \le \hat{q} \given \cD_\tr^{(1)}, \cD_{\test,j}^{(1)}\big)
& \ge \PP_{Q_{X,Y}}\big(S_{n+j} \le q^*(\Delta)\big)\\
& \ge g_{f,\rho}\big(g^{-1}_{f,\rho}(1-\alpha) - \Delta\big) \\
& \ge g_{f,\rho}\big(g_{f,\rho}^{-1}(1-\alpha)\big)
- g'_{f,\rho}\big(g^{-1}_{f,\rho}(1-\alpha)\big)\Delta,
\$
where the last step follows from the convexity of $g_{f,\rho}$ and
the separating hyperplane theorem. As proved in
Theorem~\ref{thm:cov_known_shift}, when $g_{f,\rho}(1) \ge 1-\alpha$, 
$g_{f,\rho}(g^{-1}_{f,\rho}(1-\alpha)) \ge 1-\alpha$ and we complete 
the proof.
\end{proof}

\paragraph{Proof of Theorem~\ref{thm:dr_pi}}
In this proof, we write $\Delta(\delta)$ instead of $\Delta$ to 
emphasize the dependence of $\Delta$ on $\delta$.
By Theorem~\ref{thm:hp_dr_pi}, we know that for any $\delta \in (0,1)$, 
$\PP\big(\hat{q} < q^*\big(\Delta(\delta)\big)\big) \le \delta$.

In the following, we shall consider a sequence of $\delta \in  
\{ 2^{-\ell}\}_{\ell = 0}^{\infty}$. For 
each $\ell \in \mathbb{N}$, we let $q_\ell = q^*\big(\Delta(2^{-\ell})\big)$.
\@\label{eq:peel1}
& \PP_{Q_X\times P_{Y\given X}}\big(S_{n+j} \le \hat{q}\big) - g_{f,\rho}^{-1}(1-\alpha)\\
=~& \sum^{\infty}_{\ell = 0}\EE\Big[\ind
\{q_{\ell+1} \le \hat{q} < q_{\ell} \}\cdot 
\Big(
\PP_{Q_X\times P_{Y\given X}}\big(S_{n+j} \le \hat{q} \biggiven \cD_\tr^{(k)}, 
\cD_{\test,j}^{(k)}\big) - g_{f,\rho}^{-1}(1-\alpha)\Big)\Big]\\
& \qquad \qquad + \EE\Big[ \ind\big\{\hat{q} \ge q_0\big\}  \cdot \Big(
\PP_{Q_X\times P_{Y\given X}}\big(S_{n+j} \le \hat{q} \biggiven \cD_\tr^{(k)}, 
\cD_{\test,j}^{(k)}\big) - g_{f,\rho}^{-1}(1-\alpha)\Big)\Big]\\
\ge ~&\sum^{\infty}_{\ell = 0}\EE\Big[\ind\{q_{\ell+1} \le \hat{q} < q_{\ell}\}
\cdot \big(
F({q}_{\ell+1}) - g_{f,\rho}^{-1}(1-\alpha)\big)\Big]
 + \EE\Big[ \ind\big\{\hat{q} \ge q_0\big\}  \cdot \big(
F(q_0) - g_{f,\rho}^{-1}(1-\alpha)\big)\Big]\\
\ge ~& -\sum^\infty_{\ell = 0}\Delta(2^{-\ell-1}) \cdot \PP\big( 
q_{\ell+1} \le  \hat{q} < q_{\ell}\big) - \Delta(0) \cdot \PP(\hat{q}\ge q_0)
,
\@
where the last step is due to the definition of $q_\ell$.  
Since for any $\ell \in \mathbb{N}$, $\PP(\hat{q} < q^*(\Delta(2^{-\ell}))) \le 2^{-\ell}$,
 we further have
\$
\eqref{eq:peel1}
\ge~& -\sum^\infty_{\ell = 0}\Delta(2^{-\ell})
\PP\big(\hat{q} < q^*\big(\Delta(2^{-\ell-1})\big)\big) 
- \Delta(0)\\
\ge~& -\sum^\infty_{\ell=0} 2^{-\ell-1} \Delta(2^{-\ell-1}) - \Delta(0)\\
\ge ~& 
\sup_{\kappa q^*{\bar{\Delta}} \le t \le q^*(0)}2\cdot \esterr(t)
+ \sqrt{\frac{16w_{\max}^2}{n_\tr} + \frac{2}{n_\test}}.
\$
We now return to the coverage under $Q_{X,Y}$. Again using the step 
in the proof of Theorem~\ref{thm:cov_known_shift}, we have 
\$ 
\PP_{Q_{X,Y}}(S_{n+j} \le \hat{q}) & \ge 
g_{f,\rho}\big(\PP_{Q_X \times P_{Y\given X}}(S_{n+j} \le \hat{q})\big)\\
& \ge g_{f,\rho}\bigg(g^{-1}_{f,\rho} -  
\sup_{\kappa q^*{\bar{\Delta}} \le t \le q^*(0)}2 \cdot \esterr(t)
- \sqrt{\frac{16w_{\max}^2}{n_\tr} + \frac{2}{n_\test}}\bigg)\\
& \ge g_{f,\rho}\big(g^{-1}_{f,\rho}(1-\alpha)\big)
- g'_{f,\rho}(g^{-1}_{f,\rho}(1-\alpha))\cdot 
\bigg( \sup_{\kappa q^*{\bar{\Delta}} \le t \le q^*(0)}2 \cdot \esterr(t)
+ \sqrt{\frac{16w_{\max}^2}{n_\tr} + \frac{2}{n_\test}}\bigg)
\$
When $g_{f,\rho}(1) \ge 1-\alpha$, $g_{f,\rho}(g^{-1}_{f,\rho}(1-\alpha))\ge 1-\alpha$.
The proof is thus completed.

\subsection{Proof of theorem \ref{thm:est_rho}}
Throughout, we condition on $\hat{\rho}$ and 
$\cD_{\tr}^{(1-k)}\cup \cD_{\test}^{(1-k)}$.
Fix $k \in \{0,1\}$ and $\{n+j\} \in \cI_\test^{(k)}$.
Define $A = \big\{Y_{n+j} \in \hat{C}_{f,\hat{\rho},n+j}
(X_{n+j})\big\}$.
By the proof of Theorem 1, we have that 
\begin{align}
\mathbb{P}_{(X_{n+j},Y_{n+j})\sim Q_{XY}}
(A)
&\geq  
g_{f,\rho^*}\Big(\mathbb{P}_{(X_{n+j},Y_{n+j})\sim Q_X\times P_{Y \given X}}(A )\Big).
\end{align}
Next, by the intermediate steps in the proof of Theorem~\ref{thm:cov_est}, there is
\begin{align}
P_{(X_{n+j}, Y_{n+j}) \sim Q_X \times P_{Y\given X}}(A)
&\geq P_{(X_{n+j},Y_{n+j}) \sim \tilde{Q}_X \times P_{Y  \given X}}(A)
-\frac{1}{2}\mathbb{E}_{X\sim P_X}
\bigg[\Big|\frac{\hat{w}^{(k)}(X)}{\EE_{X\sim P_X}[\hat{w}^{(k)}(X)]}-w(X)\Big|\bigg]\\
&\geq g_{f,\hat{\rho}}^{-1}(1-\alpha)
-\frac{1}{2}\mathbb{E}_{X\sim P_X}
\bigg[\Big|\frac{\hat{w}^{(k)}(X)}{\EE_{X\sim P_X}[\hat{w}^{(k)}(X)]}-w(X)\Big|\bigg].
\end{align}
Combining the above inequalities and 
since the monotonicity of $g_{f,\rho^*}(\beta)$ in $\beta$, we have 
\begin{align}
    & \mathbb{P}_{(X_{n+j},Y_{n+j})\sim Q_{XY}}\big(Y_{n+j} \in \hat{C}_{f,\hat{\rho},n+j}(X_{n+j})\big)\\
    \geq~&
    g_{f,\rho^*}\Bigg(g_{f,\hat \rho}^{-1}(1-\alpha)-
    \frac{1}{2}\mathbb{E}_{X\sim P_X}\bigg[\Big|
    \frac{\hat{w}^{(k)}(X)}{\EE_{P_X}[\hat{w}^{(k)}(X)]}-w(X)\Big|\bigg]\Bigg)\\
    \ge ~&g_{f,\rho^*}\big(g_{f,\hat \rho}^{-1}(1-\alpha)\big) 
    -\frac{1}{2} g'_{f,\rho^*}(g^{-1}_{f,\hat{\rho}}(1-\alpha))\cdot 
    \mathbb{E}_{X\sim P_X}\bigg[\Big|
    \frac{\hat{w}^{(k)}(X)}{\EE_{P_X}[\hat{w}^{(k)}(X)]}-w(X)\Big|\bigg].
\end{align}
When $\hat{\rho} \ge \rho^*$, $g_{f,\rho^*}(g^{-1}_{f,\hat{\rho}}(1-\alpha))
\ge g_{f,\hat{\rho}}(g^{-1}_{f,\hat{\rho}}(1-\alpha))$. 
The latter is greater or equal to 
$1-\alpha$ when $g_{f,\hat{\rho}}(1) \ge 1-\alpha$, following the proof of Theorem~\ref{thm:cov_known_shift}.
The proof is therefore completed.


\section{Additional results of sensitivity analysis under the $f$ sensitivity model}
\label{appx:ite}
This section collects additional results of adapting our method 
to the sensitivity analysis of ITE under the $f$-sensitivity model.

Suppose that the inferential target is $Y(t_1)$ for $t_1 \in \{0,1\}$;
and the target population is $T = t_2$, where $t_2 \in \{0,1,\circ\}$, 
with $\circ$ denoting the whole population.
The prediction interval $\hat{C}_{f,\rho}(X_{n+1})$ should satisfy 
\$ 
\PP\big(Y_{n+1}(t_1) \in \hat{C}_{f,\rho}(X_{n+1}) 
\biggiven T = t_2 \big) \ge 1-\alpha.
\$
Given a set of training data $\cD_\tr = \{(X_i, T_i,Y_i)\}_{i=1}^n$,
we start as before by randomly splitting the data into two folds 
$\cD_\tr^{(0)}$ and $\cD_\tr^{(1)}$. The first fold $\cD_\tr^{(0)}$ is 
used for fitting the propensity score function $\hat{e}$(x) (if unknown);
we also use the unit in $\cD_\tr^{(0)}$ such that $T = t_1$
to fit a function $\hat{\mu}^{(t_1)}$ for predicting $Y(t_1)$. 
The form of the covariate shift weight function  
$w^{(t_1,t_2)}(x)$ is listed in Table~\ref{tab:cov_shift}.
Next, for any $j\in [m]$, 
the prediction interval is constructed as
\begin{align}
\label{eq:ite_pi_form}
& \hat{C}_{f,\rho,n+j}^{(t_1,t_2)}(x)=
\bigg\{y\in \RR: s(x,y)\leq \text{Quantile} 
\Big(g_{f,\rho}^{-1}(1-\alpha),
\sum_{i \in \cD_\tr^{(1)},T_i = t_2} 
p_i^{(t_1,t_2)}(x) \cdot \delta_{S_i} 
+ p_{n+1}^{(t_1,t_2)}(x)  \cdot \delta_\infty \Big) \bigg\},\nonumber\\
&\text{where }p_i^{(t_1,t_2)}(x) = 
\frac{\hat{w}^{(t_1,t_2)}(X_i)}{\sum_{j \in \cD_\tr^{(1)}, T_j = t_2} 
\hat{w}^{(t_1,t_2)}(X_j) 
+ \hat{w}^{(t_1,t_2)}(x)},\nonumber \\ 
&\qquad \quad   p_{n+1}^{(t_1,t_2)}(x) = \frac{\hat{w}^{(t_1,t_2)}(x)}
{\sum_{j \in \cD_\tr^{(1)},T_j = t_2} 
\hat{w}^{(t_1,t_2)}(X_j) 
+ \hat{w}^{(t_1,t_2)}(x)}.
\end{align}
Above, $\hat{w}^{(t_1,t_2)}$ is the estimator for $w^{(t_1,t_2)}$.
The complete procedure can be found in Algorithm~\ref{alg:ite}.
\begin{table}[ht]
    \centering
    \begin{tabular}{c|c|c|c}
    \toprule
    \diagbox{$t_1$}{$t_2$}     & $1$ & $0$ & $\circ$\\
    \midrule
    $1$  & $1$ & $\tfrac{1-e(x)}{e(x)} \cdot \frac{p_1}{p_0}$ & 
    $\frac{p_1}{e(x)}$\\
    \midrule
    $0$ &$\frac{e(x)}{1-e(x)} \cdot \frac{p_0}{p_1}$ & $1$ &
    $\frac{p_0}{1-e(x)}$\\
    \bottomrule
    \end{tabular}
    \caption{The form of the covariate shift function $w^{(t_1,t_2)}(x)$. 
    The function $e(x) = \PP(T = 1\given X = x)$ is the observed propensity
    score function, $p_1 = \PP(T=1)$, and $p_0 = \PP(T=0)$.}
    \label{tab:cov_shift}
\end{table}

\begin{algorithm}[htbp]
\caption{Conformalized counterfactual inference under the $f$-sensitivity model}
\label{alg:ite}
\KwIn{Training set $\cD_\tr = \{(X_i,T_i,Y_i)\}_{i=1}^n$;
test data $\{X_{n+j}\}_{j=1}^m$\;
counterfactual type $t_1 \in \{0,1\}$, target population $t_2 \in \{0,1,\circ\}$\; 
outcome fitting algorithm $\mathcal{A}$; propensity score fitting 
algorithm $\cE$\;
target miscoverage level $\alpha \in (0,1)$; score function $s(x,y;\mu)$;
sensitivity parameter $\rho$.}

\vskip 1em
Randomly split $\cD_\tr$ into two subsets of equal sizes 
$\mathcal{D}_\tr^{(0)}$ and  $\mathcal{D}_\tr^{(1)}$, indexed by 
$\cI_\tr^{(0)}$ and $\cI_\tr^{(1)}$\;
\vskip 0.2em 
Apply $\mathcal{A}$ to obtain the outcome regression functions:  
$\hat{\mu}^{(t_1)} \leftarrow \cA\big(\{(X_i,Y_i): i \in \cI_\tr^{(0)},T_i = t_1\}\big)$\;
\vskip 0.2em 
Apply $\cE$ to obtain the estimated propensity score function: 
$\hat{e}\leftarrow \mathcal{E}(\cD_\tr^{(0)})$\;
\vskip 0.2em 
Compute the nonconformity scores:
$S_i= s(X_i, Y_i; \hat{\mu}^{(t_1)})$ 
for $i\in \mathcal{I}_\tr^{(1)}$ such that $T_i = t_1$\;
\vskip 0.4em
\For{$\ell \in \cI_\test$}{
Construct $\hat{C}^{(t_1,t_2)}_{f,\rho,\ell}(X_\ell)$ according to~\eqref{eq:ite_pi_form}\;
}

\KwOut{Prediction sets $\{\hat{C}^{(t_1,t_2)}_{f,\rho,n+j}(X_{n+j})\}_{j \in [m]}$.}
\end{algorithm}

\section{Additional simulation results}
\label{appx:simulation}

This section collects results from additional simulation studies. 
In specific, we evaluate the candidate methods under a pure $Y \given X$ shift setting, 
a pure $X$ shift setting, and a nonlinear model setting; the effect of the choice of $f$ 
is also studied.

\subsection{Pure $Y\given X$ Shift.}
We consider a setting where the distributional shift is purely in the 
conditional relationship between $Y$ and $X$ while the covariate distribution 
remains invariant, i.e. $P_X  =  Q_X$. The other settings are the same as 
in Section~\ref{sec:simulation}. 
The corresponding results are presented in Figure~\ref{fig:sim_no_cov}. 
Since the distributional shift is purely in $Y\given X$, \texttt{WCP} suffers 
the same degree of coverage drop as \texttt{CP}. \texttt{WRCP} and \texttt{D-WRCP}
provide valid and tight prediction intervals. The prediction interval of 
\texttt{RCP} is overly conservative due to the estimation error of 
the $X$ shift.

\begin{figure}[h]
\centering
\includegraphics[width = 0.5\textwidth]{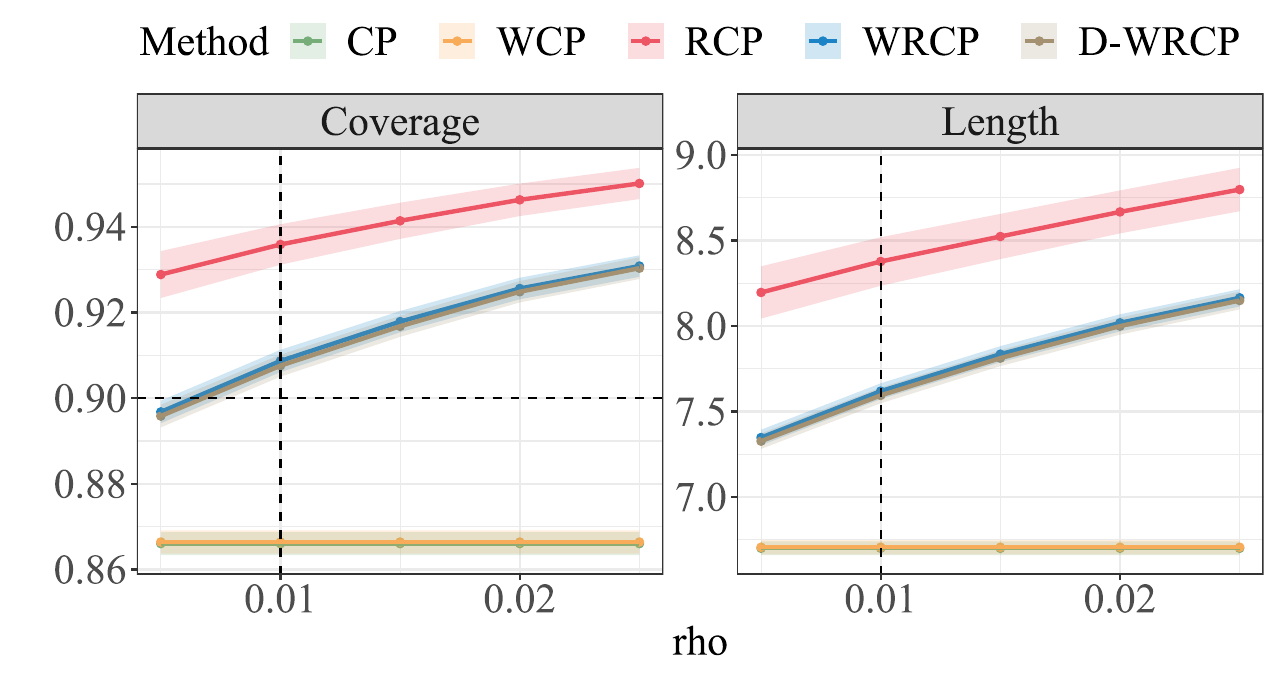}
\caption{Averaged coverage (left) and prediction interval length (right) under pure $Y\given X$ shift. 
The details are otherwise the same as in Figure~\ref{fig:sim-res}.} 
\label{fig:sim_no_cov}
\end{figure}

\subsection{Pure $X$ Shift.}
We now consider a pure covariate shift, 
fixing the conditional relationship between $Y$ and $X$. 
The other setup is the same as in Section~\ref{sec:simulation}, 
except that we take $\eta$ (the $X$ shift parameter) to be $0.1$.
The results are shown in Figure~\ref{fig:sim_no_cond}, where 
\texttt{WCP} achieves the desired coverage level; our proposed 
method is a bit conservative, but less severe than \texttt{RCP}.

\begin{figure}[h]
\centering
\includegraphics[width = 0.5\textwidth]{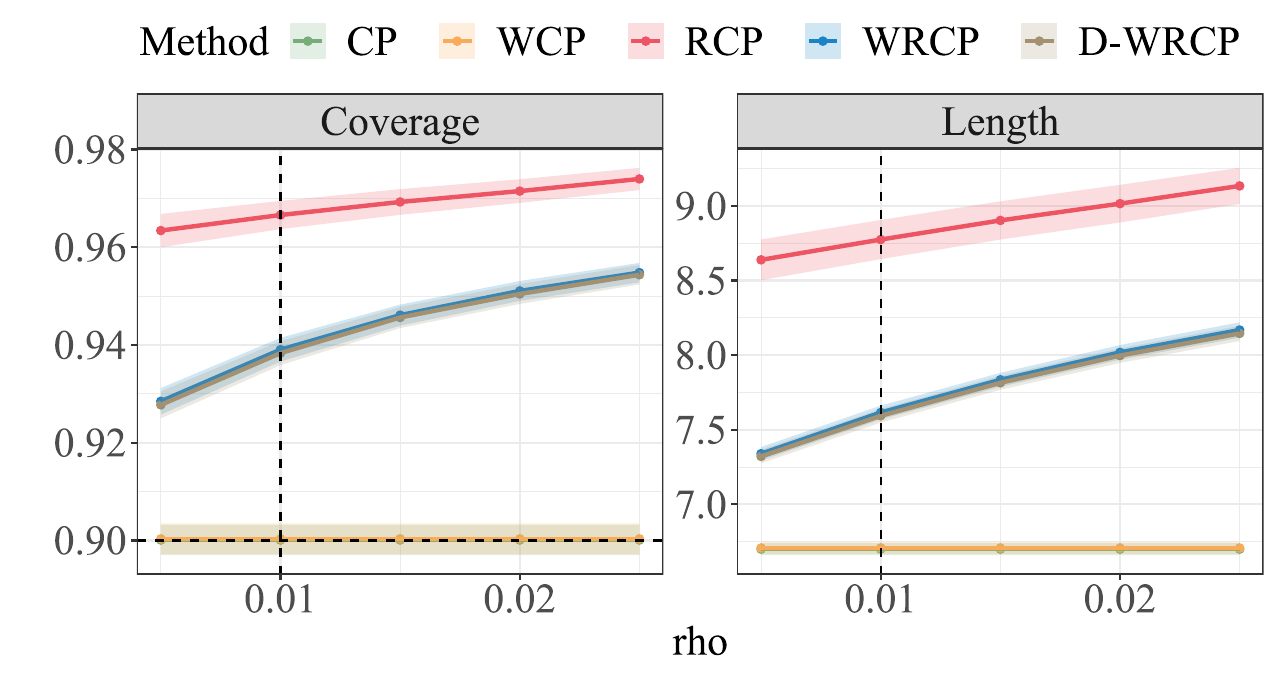}
\caption{Averaged coverage (left) and prediction interval length (right) under pure $X$ shift. 
The details are otherwise the same as in Figure~\ref{fig:sim-res}.} 
\label{fig:sim_no_cond}
\end{figure}

\subsection{Nonlinear Models.}
Next, we consider a nonlinear relationship between $Y$ and $X$. 
Under $P$,  
\$ 
Y = \frac{1}{(1+e^{X_1})(1+e^{-X_2})} + \eps, \text{ where }\eps \sim \cN(0,1).
\$
The $X$ shift parameter $\eta$ is set to be $0.1$, and all the other 
settings are the same as in Section~\ref{sec:simulation}.
The corresponding results are presented in Figure~\ref{fig:sim_nonlinear}, 
where the message is similar to the linear case: \texttt{WCP} is not able to  
address the distributional shift,  \texttt{RCP} is overly conservative,
while our proposed methods provide valid and tight prediction intervals.

\begin{figure}[h]
\centering
\includegraphics[width = 0.5\textwidth]{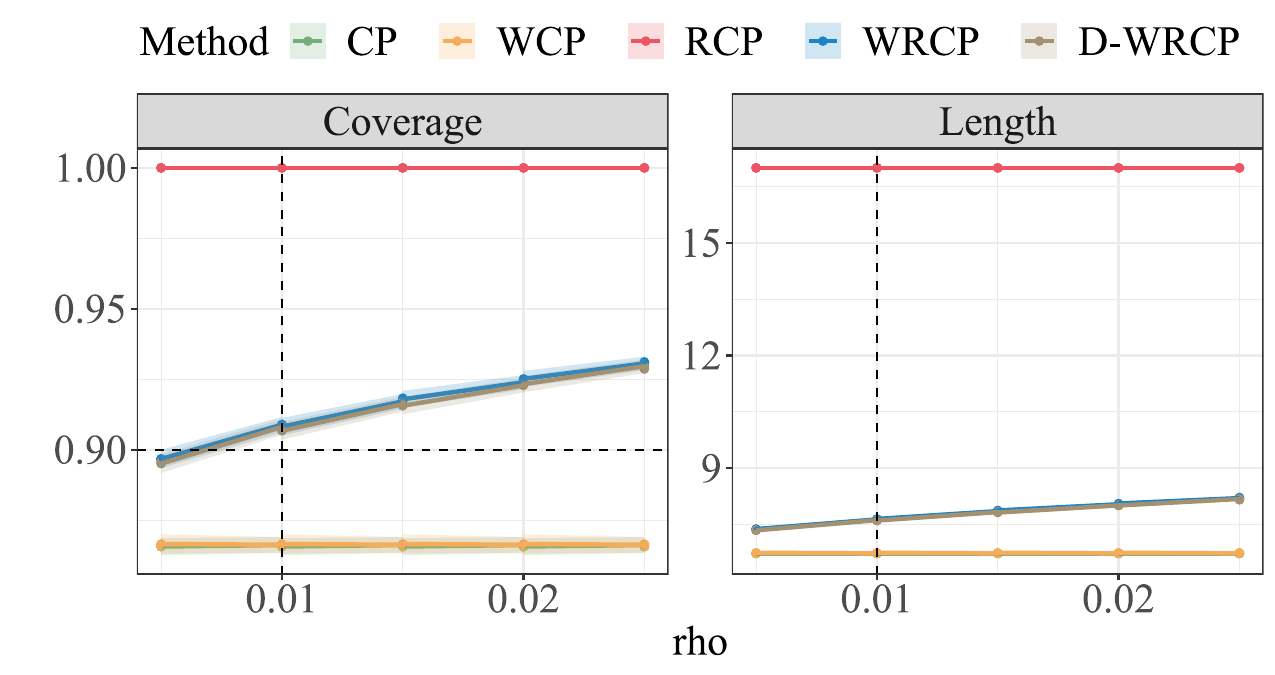}
\caption{Averaged coverage (left) and prediction interval length (right) 
under the nonlinear model experiment. 
The details are otherwise the same as in Figure~\ref{fig:sim-res}.} 
\label{fig:sim_nonlinear}
\end{figure}

\subsection{Effect of $f$.}
Finally, we investigate the effect of the choice of $f$ on the performance of 
our proposed methods. We consider the same setup as in Section~\ref{sec:simulation}
with $\rho^* = 0.01$ (under KL divergence) and $\eta = 0.5$. 
We vary the choice of $f$ in the implementation of \texttt{WRCP} to be 
(1) $f(x) = x\log x$ (KL divergence), (2) $f(x) = \frac{1}{2}|x-1|$ (TV distance), 
and (3) $f(x) = (x-1)^2$ ($\chi^2$ divergence).
Figure~\ref{fig:f_comp} demonstrates the averaged coverage and prediction interval length 
under different choices of $f$. The results show that the choice of $f$ has a rather 
mild effects on the performance of \texttt{WRCP}.

\begin{figure}[h]
\centering
\includegraphics[width = 0.5\textwidth]{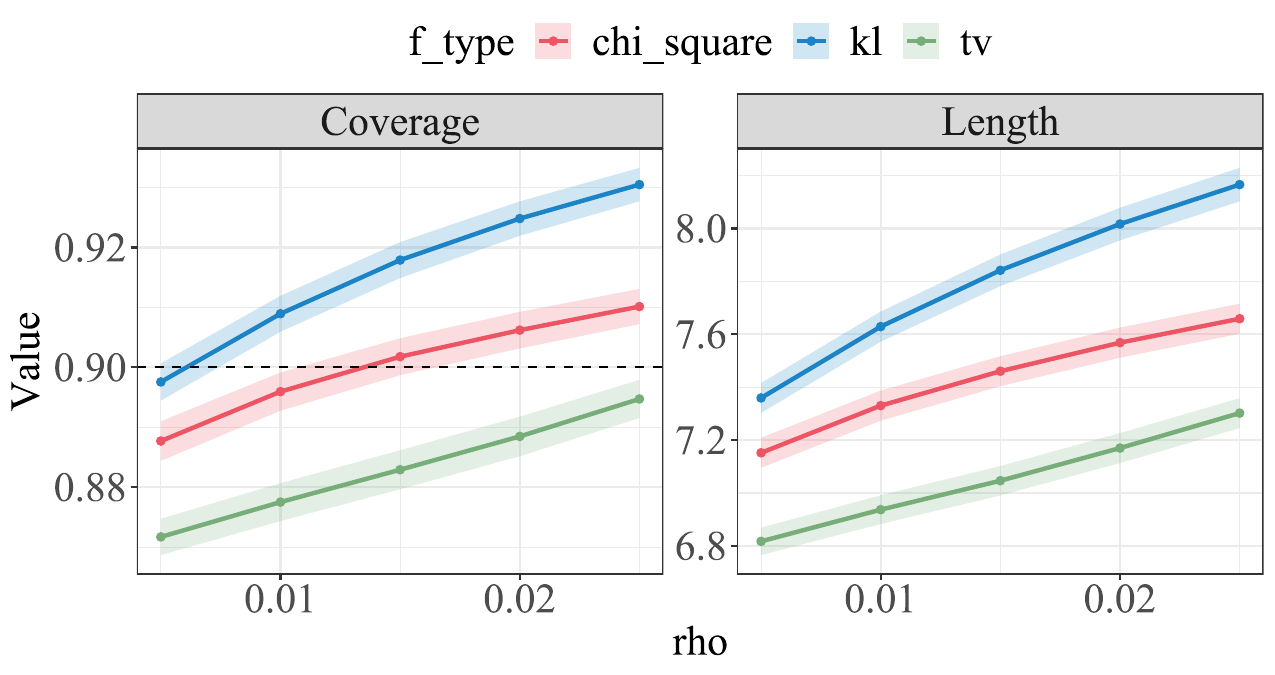}
\caption{Averaged coverage (left) and prediction interval length (right) 
with different choices of $f$. 
The details are otherwise the same as in Figure~\ref{fig:sim-res}.} 
\label{fig:f_comp}
\end{figure}

\section{Additional Real Data Results}
\label{appx:real_data}
We additionally evaluate the candidate methods on the poverty mapping dataset~\citep{yeh2020using},
    where the task is to estimate the poverty rate in different regions of the world 
    with the help of satellite images. The data is obtained and preprocessed with  
    the WILDS python package~\citep{wilds2021}.

    The covariate $X$ corresponds to 
    a $224\times 224$-pixel image, and the outcome $Y$ is  a real-valued asset wealth index
    (computed from Demographic and Health Surveys data). In our implementation, 
    the training set consists of $500$ samples, and the test set consists of $200$ samples.
    We fit a convolutional neural network (CNN) for $\hat \mu$; for $\hat w$, the 
    variational autoencoder (VAE) is used to learn a representation of the images, upon 
    which a random forest classifier is trained to estimate the $X$ shift.

    We consider $\rho \in \{0.1,0.12,\ldots,0.18\}$, repeating the above process for 100 random splits 
    under each $\rho$. Figure~\ref{fig:poverty} shows the results of \texttt{CP}, \texttt{WCP}, 
    \texttt{RCP}, and \texttt{WRCP}. We again see that \texttt{CP} and \texttt{WCP} show 
    significant under-coverage, while \texttt{RCP} is overly conservative. \texttt{WRCP} 
    delivers valid and tight prediction intervals for a wide range of $\rho$'s.

    \begin{figure}[ht]
        \centering
        \includegraphics[width = 0.7\textwidth]{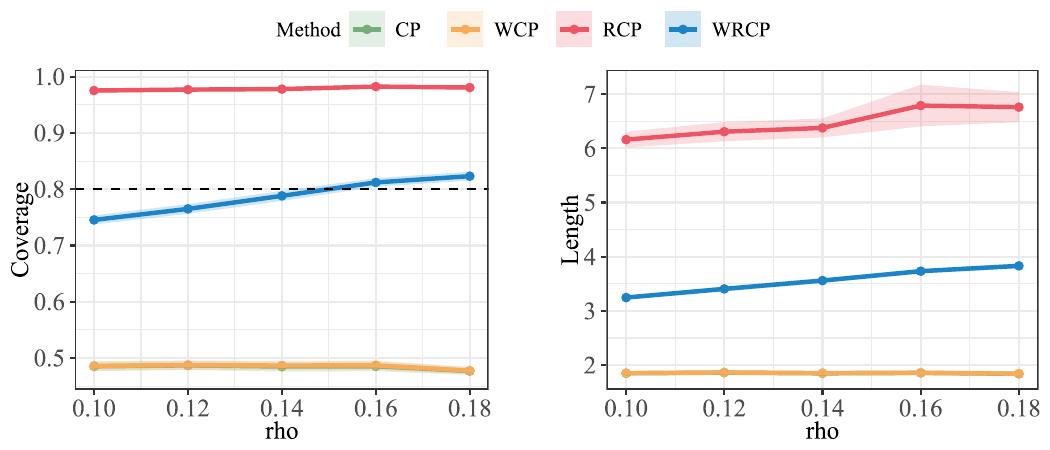}
        \caption{Averaged coverage (left) and prediction interval length (right)
        over $100$ runs as a function of the robust parameter $\rho$
        from the experiment on poverty estimation. 
        The other details are the same as in Figure~\ref{fig:nslm}.}
        \label{fig:poverty}
    \end{figure}
\end{document}